\documentclass[12pt,onecolumn,journal]{IEEEtran}

\ifCLASSINFOpdf

\else

\fi
\usepackage[normalem]{ulem}
\usepackage{cite}
\usepackage{amsmath}
\usepackage{amssymb}
\usepackage{amsfonts}
\usepackage{amsthm}
\usepackage{mathtools}
\usepackage{graphicx}
\usepackage{color}
\usepackage{epstopdf}
\usepackage{caption}
\usepackage{subfigure}
\usepackage[numbers,sort&compress]{natbib}
\usepackage[left=2cm, right=2cm, top=2.54cm]{geometry}
\DeclareMathOperator{\Var}{\operatorname{Var}}

\newcommand{\norm}[1]{\left\lVert#1\right\rVert}
\newcommand{\stkout}[1]{\ifmmode\text{\sout{\ensuremath{#1}}}\else\sout{#1}\fi}

\usepackage{xcolor,cancel}

\newtheorem{thm}{Theorem}

\newtheorem{lem}{Lemma}

\newtheorem{assum}{Assumption}

\newtheorem{defi}{Definition}
\begin{document}
	
	\title{Minimax Optimal Estimation of KL Divergence for Continuous Distributions}
	
	\author{Puning Zhao and
		Lifeng Lai\thanks{Puning Zhao and Lifeng Lai are with Department of Electrical and Computer Engineering, University of California, Davis, CA, 95616. Email: \{pnzhao,lflai\}@ucdavis.edu. This work was supported by the National Science Foundation under grants CCF-17-17943, ECCS-17-11468, CNS-18-24553 and CCF-19-08258. }
	}
	\maketitle
	
	\begin{abstract}
Estimating Kullback-Leibler divergence from identical and independently distributed samples is an important problem in various domains. One simple and effective estimator is based on the $k$ nearest neighbor distances between these samples. In this paper, we analyze the convergence rates of the bias and variance of this estimator. Furthermore, we derive a lower bound of the minimax mean square error and show that kNN method is asymptotically rate optimal.
	\end{abstract}

\section{Introduction}
Kullback-Leibler (KL) divergence has a broad range of applications in information theory, statistics and machine learning. For example, KL divergence can be used in hypothesis testing \cite{anderson1994two}, text classification \cite{dhillon2003divisive}, outlying sequence detection \cite{bu2016universal}, multimedia classification \cite{moreno2004kullback}, speech recognition \cite{ramirez2004new}, etc. In many applications, we hope to know the value of KL divergence, but the distributions are unknown. Therefore, it is important to estimate KL divergence based only on some identical and independently distributed (i.i.d) samples. Such problem has been widely studied \cite{bu2018estimation,wang2005divergence,wang2009divergence,antos2001convergence,paul2019practical,cai2002universal,cai2006universal,zhang2014nonparametric}.

The estimation method is different depending on whether the underlying distribution is discrete or continuous. For discrete distributions, an intuitive method is called plug-in estimator, which first estimates the probability mass function (PMF) by simply counting the number of occurrences at each possible value and then calculates the KL divergence based on the estimated PMF. However, since it is always possible that the number of occurrences at some locations is zero, this method has infinite bias and variance for arbitrarily large sample size. As a result, it is necessary to design some new estimators, such that both the bias and variance converge to zero. Several methods have been proposed in \cite{zhang2014nonparametric,cai2002universal,cai2006universal}. These methods perform well for distributions with fixed alphabet size. Recently, there is a growing interest in designing estimators that are suitable for distributions with growing alphabet size. \cite{bu2018estimation} provided an `augmented plug-in estimator', which is a modification of the simple plug-in method. The basic idea of this method is to add a term to both the numerator and the denominator when calculating the ratio of the probability mass. Although this modification will introduce some additional bias, the overall bias is reduced. Moreover, a minimax lower bound has also been derived in \cite{bu2018estimation}, which shows that the augmented plug-in estimator proposed in~\cite{bu2018estimation} is rate optimal.

For continuous distributions, there are also many interesting methods. A simple one is to divide the support into many bins, so that continuous values can be quantized, and then the distribution can be converted to a discrete one. As a result, the KL divergence can be estimated based on these two discrete distributions. However, compared with other methods, this method is usually inefficient, especially when the distributions have heavy tails, as the probability mass of a bin at the tail of distributions is hard to estimate. 
An improvement was proposed in \cite{wang2005divergence}, which is based on data dependent partitions on the densities with an appropriate bias correction technique. Comparing with the direct partition method mentioned above, this adaptive one constructs more bins at the regions with higher density, and vice versa, to ensure that the probability mass in each bins are approximately equal. It is shown in \cite{wang2005divergence} that this method is strongly consistent. Another estimator was designed in \cite{nguyen2010estimating}, which uses a kernel based approach to estimate the density ratio. There are also some previous works that focus on a more general problem of estimating $f$-divergence, with KL divergence being a special case. For example, \cite{moon2014ensemble} constructed an estimator based on a weighted ensemble of plug in estimators, and the parameters need to be tuned properly to get a good bias and variance tradeoff. Another method of estimating $f$-divergence in general was proposed in \cite{paul2019practical}, under certain structural assumptions.

Among all the methods for the estimation of KL divergence between two continuous distributions, a simple and effective one is $k$ nearest neighbor (kNN) method based estimator. kNN method, which was first proposed in \cite{fix1951discriminatory}, is a powerful tool for nonparametric statistics. Kozachenko and Leonenko \cite{kozachenko1987sample} designed a kNN based method for the estimation of differential entropy, which is convenient to use and does not require too much parameter tuning. Both theoretical analysis and numerical experiments show that this method has desirable accuracy \cite{tsybakov1996root,singh2016analysis,singh2016finite,gao2018demystifying,berrett2019efficient,zhao2019analysis,khan2007relative}. In particular, \cite{zhao2019analysis} shows that this estimator is nearly minimax rate optimal under some assumptions. The estimation of KL divergence shares some similarity with that of entropy estimation, since KL divergence between $f$ and $g$, which denotes the probability density functions (pdf) of two distributions, is actually the difference of the entropy of $f$ and the cross entropy between $f$ and $g$. As a result, the idea of Kozachenko-Leonenko entropy estimator can be used to construct a kNN based estimator for KL divergence, which was first proposed in \cite{wang2009divergence}. The basic idea of this estimator \cite{wang2009divergence} is to obtain an approximate value of the ratio between $f$ and $g$ based on the ratio of kNN distances. It has been discussed in \cite{wang2009divergence} that, compared with other KL divergence estimators, the kNN based estimator has a much lower sample complexity, and is easier to generalize and implement for high dimensional data. Moreover, it was proved in \cite{wang2009divergence} that the kNN based estimator is consistent, which means that both the bias and the variance converge to zero as sample sizes increase. However, the convergence rate remains unknown.

In this paper, we make the following two contributions. Our first main contribution is the analysis of the convergence rates of bias and variance of the kNN based KL divergence estimator proposed in \cite{wang2009divergence}. For the bias, we discuss two significantly different types of distributions separately. In the first type of distributions analyzed, both $f$ and $g$ have bounded support, and are bounded away from zero. One such example is when both distributions are uniform distributions. This implies that the distribution has boundaries, where the pdf suddenly changes. There are two main sources of estimation bias of kNN method for this case. The first source is the boundary effect, as the kNN method tends to underestimate the pdf values at the region near the boundary. The second source is the local non-uniformity of the pdf. It can be shown that the bias caused by the second source converges fast enough and thus can be negligible. As a result, the boundary bias is the main cause of bias of the kNN based KL divergence estimator for the first type of distributions considered. In the second type of distributions analyzed, we assume that both $f$ and $g$ are continuous everywhere. For example, a pair of two Gaussian distributions with different mean or variance belong to this case. For this type of distributions, the boundary effect does not exist. However, as the density values can be arbitrarily close to zero, we need to consider the bias caused by the tail region, in which $f$ or $g$ is too low and thus kNN distances are too large for us to obtain an accurate estimation of the density ratio $f/g$. For the variance of this estimator, we bound the convergence rate under a unified assumption, which holds for both two cases discussed above. The convergence rate of the mean square error can then be obtained based on that of the bias and variance. In this paper, we assume that $k$ is fixed. We will show that with fixed $k$, the convergence rate of the mean square error over the sample sizes is already minimax optimal.

Our second main contribution is to derive a minimax lower bound of the mean square error of KL divergence estimation, which characterizes the theoretical limit of the convergence rates of any methods. For discrete distributions, the minimax lower bound has already been derived in \cite{han2016minimax} and \cite{bu2018estimation}. However, for continuous distributions, the minimax lower bound has not been established. In fact, there exists no estimators that are uniformly consistent for all continuous distributions. For example, let $f=\sum_{i=1}^m p_i\mathbf{1}((i-1)/m<x\leq i/m)$, in which $\mathbf{1}$ is the indicator function, and $g$ is uniform in $[0,1]$. Then the estimation error of KL divergence between $f$ and $g$ equals the estimation error of the entropy of $\mathbf{p}=(p_1,\ldots,p_m)$. Since $m$ can be arbitrarily large, according to the lower bound derived in \cite{wu2016minimax}, there exists no uniformly consistent estimator. As a result, to find a minimax lower bound, it is necessary to impose some restrictions on the distributions. In this paper, we analyze the minimax lower bound for two cases that match our assumptions for deriving the upper bound, i.e. distributions with bounded support and densities bounded away from zero, and distributions that are smooth everywhere and densities can be arbitrarily close to zero. For each case, we show that the minimax lower bound nearly matches our upper bound using kNN method. This result indicates that the kNN based KL divergence estimator is nearly minimax optimal. To the best of our knowledge, our work is the first attempt to analyze the convergence rate of KL divergence estimator based on kNN method, and prove its minimax optimality. 

The remainder of this paper is organized as follows. In Section \ref{sec:main}, we provide the problem statements. In Sections \ref{sec:bias} and \ref{sec:variance}, we characterize the convergence rates of the bias and variance of the kNN based KL divergence estimator respectively. In Section \ref{sec:mmx}, we show the minimax lower bound. We then provide numerical examples in Section \ref{sec:numerical}, and concluding remarks in Section \ref{sec:conclusion}.

\section{Problem Statement}\label{sec:main}
Consider two pdfs $f,g:\mathbb{R}^d\rightarrow \mathbb{R}$ where $f(\mathbf{x})>0$ only if $g(\mathbf{x})>0$. The KL divergence between $f$ and $g$ is defined as
\begin{eqnarray}
D(f||g)=\int f(\mathbf{x})\ln \frac{f(\mathbf{x})}{g(\mathbf{x})}d\mathbf{x}.
\label{eq:kldef}
\end{eqnarray}
$f$ and $g$ are unknown. However, we are given a set of samples $\{\mathbf{X}_1,\ldots,\mathbf{X}_N\}$ drawn i.i.d from pdf $f$, and another set of samples $\{\mathbf{Y}_1,\ldots,\mathbf{Y}_M \}$ drawn i.i.d from pdf $g$. The goal is to estimate $D(f||g)$ based on these samples.

\cite{wang2009divergence} proposed a kNN based estimator:
\begin{eqnarray}
\hat{D}(f||g)=\frac{d}{N}\sum_{i=1}^N\ln \frac{\nu_i}{\epsilon_i}+\ln \frac{M}{N-1},
\label{eq:estimator}
\end{eqnarray}
in which $\epsilon_i$ is the distance between $\mathbf{X}_i$ and its $k$-th nearest neighbor in $\{\mathbf{X}_1,\ldots, \mathbf{X}_{i-1},\mathbf{X}_{i+1},\ldots,\mathbf{X}_N \}$, while $\nu_i$ is the distance between $\mathbf{X}_i$ and its $k$-th nearest neighbor in $\{\mathbf{Y}_1,\ldots, \mathbf{Y}_M \}$, $d$ is the dimension. The distance between any two points $\mathbf{u},\mathbf{v}$ is defined as $\norm{\mathbf{u}-\mathbf{v}}$, in which $\norm{\cdot}$ can be an arbitrary norm. The basic idea of this estimator is using kNN method to estimate the density ratio. An estimation of $f$ at $\mathbf{X}_i$ is
\begin{eqnarray}
\hat{f}(\mathbf{X}_i)=\frac{k}{N-1}\frac{1}{V(B(\mathbf{X}_i,\epsilon_i))},
\label{eq:fhat}
\end{eqnarray}
in which $V(S)$ is the volume of set $S$. \eqref{eq:fhat} can be understood as follows. Apart from $\mathbf{X}_i$, there are another $N-1$ samples from $\mathbf{X}_1,\ldots, \mathbf{X}_N$, among which $k$ points fall in $V(B(\mathbf{X}_i,\epsilon_i))$. Therefore, $k/(N-1)$ is an estimate of $P_f(B(\mathbf{X}_i,\epsilon_i))$, in which $P_f$ is the probability mass with respect to the distribution with pdf $f$. As the distribution is continuous, we have $P_f(B(\mathbf{X}_i,\epsilon_i))\approx f(\mathbf{X}_i)V(B(\mathbf{X}_i,\epsilon_i))$. We can then use \eqref{eq:fhat} to estimate $\hat{f}(\mathbf{X}_i)$. Similarly, as there are $M$ samples $\mathbf{Y}_1,\ldots, \mathbf{Y}_M$ generated from $g$, we can obtain an estimate $\hat{g}$ by
\begin{eqnarray}
\hat{g}(\mathbf{X}_i)=\frac{k}{M}\frac{1}{V(B(\mathbf{X}_i,\nu_i))}.
\label{eq:ghat}
\end{eqnarray}
As
\begin{eqnarray}
D(f||g)=\mathbb{E}_{\mathbf{X}\sim f}\left[\ln \frac{f(\mathbf{X})}{g(\mathbf{X})} \right]\approx \frac{1}{N}\sum_{i=1}^N \ln \frac{f(\mathbf{X}_i)}{g(\mathbf{X}_i)},
\end{eqnarray}
by replacing $f(\mathbf{X}_i)$, $g(\mathbf{X}_i)$ with \eqref{eq:fhat} and \eqref{eq:ghat} respectively, we can get the expression of the KL divergence estimator in \eqref{eq:estimator}. 

\cite{wang2009divergence} has proved that this estimator is consistent, but the convergence rate remains unknown. In this paper, we analyze the convergence rates of the bias and variance of this estimator, and derive the minimax lower bound.

\section{Bias Analysis}\label{sec:bias}

In this section, we derive convergence rate of the bias of the estimator~\eqref{eq:estimator}. We will consider two different cases depending on whether the support is bounded or not, as they have different sources of biases.

\subsection{The Case with Bounded Support}

We first discuss the case in which the distributions have bounded support and the densities are bounded away from zero. The main source of bias of this case is boundary effects. The analysis is based on the following assumptions:

\begin{assum}\label{ass:bounded}
	Assume the following conditions: 
	
	(a) $S_f\subset S_g$, in which $S_f$ and $S_g$ are the supports of $f$ and $g$;
	
	(b) There exist constants $L_f$, $U_f$, $L_g$, $U_g$ such that $L_f\leq f(\mathbf{x})\leq U_f$ for all $\mathbf{x}\in S_f$ and $L_g\leq g(\mathbf{x})\leq U_g$ for all $\mathbf{x}\in S_g$;
	
	(c) The surface areas (or Hausdorff measure) of $S_f$ and $S_g$ are bounded by $H_f$ and $H_g$;
	
	(d) The diameters of $S_f$ and $S_g$ are bounded by $R$, i.e. $\underset{\mathbf{x}_1,\mathbf{x}_2\in S_g}{\sup} \norm{\mathbf{x}_2-\mathbf{x}_1}<R$;
	
	(e) There exists a constant $0<a<1$ such that for all $r\leq R$ and $\mathbf{x}\in S_f$, $V(B(\mathbf{x},r)\cap S_f)\geq a V(B(\mathbf{x},r))$, and for all $\mathbf{x}\in S_g$, $V(B(\mathbf{x},r)\cap S_g)\geq a V(B(\mathbf{x},r))$, in which $V$ denotes the volume of a set;
	
	(f) The Hessian of $f$ and $g$ are both bounded by $C_0$.
	
\end{assum}
Assumption (a) is necessary to ensure that the definition of KL divergence in \eqref{eq:kldef} is valid. (b) bounds both the lower and upper bound of the pdf value. (c) restricts the surface area of the supports of $f$ and $g$. Since the kNN divergence estimator tends to cause significant bias at the region near to the boundary, the estimation bias for distributions with irregular supports with large surface area are usually large. (d) requires the boundedness of the support. The case with unbounded support will be considered in Section~\ref{sec:biasub}. (e) ensures that the angles at the corners of the support sets have a lower bound, so that there will not be significant bias at the corner region. (f) ensures the smoothness of distribution in the support set. Note that \eqref{eq:fhat} and \eqref{eq:ghat} actually estimate the average density $f$ and $g$ over the ball $B(\mathbf{X}_i,\epsilon_i)$ and $B(\mathbf{X}_i,\nu_i)$. If the $f$ and $g$ are smooth, then the average values will not deviate too much from the pdf value at the center of the balls, i.e. $f(\mathbf{X}_i)$ and $g(\mathbf{X}_i)$. 

Based on the above assumptions, we have the following theorem regarding the bias of estimator~\eqref{eq:estimator}.
\begin{thm}\label{thm:bounded}
	Under Assumption \ref{ass:bounded}, the convergence rate of the bias of kNN based KL divergence estimator is bounded by:
	\begin{eqnarray}
	|\mathbb{E}[\hat{D}(f||g)]-D(f||g)|=\mathcal{O}\left(\left(\frac{\ln \min\{M,N \}}{\min\{M,N\}}\right)^\frac{1}{d}\right).
	\label{eq:bounded}
	\end{eqnarray}
\end{thm}
\begin{proof}
	(Outline) Considering that
	\begin{eqnarray}
	D(f||g)=-h(f)-\int f(\mathbf{x})\ln g(\mathbf{x})d\mathbf{x},
	\end{eqnarray}
	in which $h$ denotes the differential entropy, we decompose the KL divergence estimator to an estimator of the differential entropy of $f$, as well as an estimator of the cross entropy between $f$ and $g$. We then bound the bias of these two estimators. In particular, we can write
	\begin{eqnarray}
	\mathbb{E}[\hat{D}(f||g)]-D(f||g)=-I_1+I_2+I_3,\label{eq:dico}
	\end{eqnarray}
	with
	\begin{eqnarray}
	I_1&=&-\psi(k)+\psi(N)+\ln c_d+d\mathbb{E}[\ln \epsilon]-h(f),\nonumber\\
	I_2&=&-\psi(k)+\psi(M+1)+\ln c_d+d\mathbb{E}[\ln \nu]+\mathbb{E}[\ln g(\mathbf{X})],\nonumber\\
	I_3&=&\ln M-\psi(M+1)-\ln(N-1)+\psi(N),	
	\end{eqnarray}
	in which $\psi$ is the digamma function, $\psi(u)=d(\ln \Gamma(u))/du$, with $\Gamma$ being the Gamma function. Due to the property of Gamma distribution, we know that $|\ln M-\psi(M+1)|\leq 1/M$, and $|\ln (N-1)-\psi(N)|\leq 1/N$. Hence $I_3$ decays sufficiently fast and can be negligible for large sample sizes $N$ and $M$.
	
	$I_1$ has the same form as the bias of Kozachenko-Leonenko entropy estimator \cite{kozachenko1987sample}, which has been analyzed in many previous literatures \cite{biau2015lectures,gao2018demystifying,singh2016analysis,berrett2019efficient,zhao2019analysis}. With some modifications, the proofs related to the entropy estimator can also be used to bound $I_2$, which is actually the bias of a cross entropy estimator. However, as the assumptions are different from the assumptions made in previous literatures, we need to derive \eqref{eq:bounded} in a different way.
	
	In our proof, for both the entropy estimator and the cross entropy estimator, we divide the support into two parts, the central region and the boundary region. In the central region, $B(\mathbf{x},\epsilon)$ will be within $S_f$ and $B(\mathbf{x},\nu)$ will be within $S_g$ with high probability. Since $f$ and $g$ are smooth, the expected estimate $\hat{f}$ and $\hat{g}$ are very close to the truth, and thus will not cause significant bias. The main bias comes from the boundary region, in which the density estimator $\hat{f}$ and $\hat{g}$ are no longer accurate, as $B(\mathbf{x},\epsilon)$ or $B(\mathbf{x},\nu)$ exceeds the supports $S_f$ and $S_g$. We bound the boundary bias by letting the boundary region to shrink with a proper speed. 
	
	The detailed proof is shown in Appendix \ref{sec:bounded}.
\end{proof}

\subsection{The Case with Smooth Distributions}\label{sec:biasub}
We now consider the second case where the density is smooth everywhere and the density can be arbitrarily close to zero. For this case, the main source of bias is tail effects. We make the following assumptions:
\begin{assum}\label{ass:unbounded}
	Assume the following conditions: 
	
	(a) If $f(\mathbf{x})>0$, then $g(\mathbf{x})>0$;
	
	(b) $\text{P}(f(\mathbf{X})\leq t)\leq \mu t^\gamma$ and $\text{P}(g(\mathbf{X})\leq t)\leq \mu t^\gamma$ for some constants $\mu$ and $\gamma\in (0,1]$, in which $\mathbf{X}$ follows a distribution with pdf $f$;
	
	(c)$\norm{\nabla^2 f}\leq C_0$, $\norm{\nabla^2 g}\leq C_0$ for some constant $C_0$;
	
	(d) $\mathbb{E}[\norm{\mathbf{X}}^s]\leq K$, and $\mathbb{E}[\norm{\mathbf{Y}}^s]\leq K$ for some constants $s>0$, $K>0$.
\end{assum}

Assumption (a) ensures that the definition of KL divergence in \eqref{eq:kldef} is valid. (b) is the tail assumption. A lower $\gamma$ indicates a stronger tail, and thus the convergence of bias of the KL divergence estimator will be slower. For example, $\gamma=1$ for Gaussian distribution and $\gamma=1/2$ for Cauchy distribution. (c) is the smoothness assumption. (d) is an additional tail assumption, which is actually very weak and holds for almost all of the common distributions, since $s$ can be arbitrarily small. However, this assumption is important since it prevents very large $\epsilon$ and $\nu$. Based on the above assumptions, we have the following theorem regarding the bias of estimator~\eqref{eq:estimator}.

\begin{thm}\label{thm:unbounded}
	Under Assumption \ref{ass:unbounded}, the convergence rate of the bias of kNN based KL divergence estimator is bounded by:
	\begin{eqnarray}
\left|\mathbb{E}[\hat{D}(f||g)]-D(f||g)\right|=\mathcal{O}\left((\min\{M,N \})^{-\frac{2\gamma}{d+2}}\ln \min\{M,N \}\right).
	\label{eq:unbounded}
	\end{eqnarray}
\end{thm}
\begin{proof}
	(Outline) Similar to the proof of Theorem \ref{thm:bounded}, we still decompose the KL divergence estimator to two estimators that estimate the entropy of $f$ and the cross entropy between $f$ and $g$, separately. In particular, we can still decompose the bias using~\eqref{eq:dico}. For simplicity, we only provide the convergence bound of $I_2$, which is the error of the cross entropy estimator. The bound of the entropy estimator holds similarly. 
	
	For the cross entropy estimator, we divide the support into two parts, including a central region $S_1$, in which $f$ or $g$ is relatively high, and a tail region $S_2$, in which $f$ or $g$ is relatively low. According to the results of order statistics \cite{david1970order,biau2015lectures}, $\mathbb{E}[\ln P_g(B(\mathbf{x},\nu))]=\psi(k)-\psi(M+1)$, in which $P_g(S)$ is the probability mass of $S$ with respect to the distribution with pdf $g$. Therefore, $I_2$ can be bounded by
	\begin{eqnarray}
	|I_2|&=&\left|\mathbb{E}\left[\ln \frac{P_g(B(\mathbf{X},\nu))}{c_d\nu^d g(\mathbf{X})}\right]\right|\leq \sum_{i=1}^2\left|\mathbb{E}\left[\ln \frac{P_g(B(\mathbf{X},\nu))}{c_d\nu^d g(\mathbf{X})}\mathbf{1}(\mathbf{X}\in S_i)\right]\right|.
	\label{eq:I2decomp}
	\end{eqnarray} 
	
	We bound two terms in \eqref{eq:I2decomp} separately. To derive the bound of bias in $S_1$, we find a high probability upper bound of $\nu_i$, denoted as $\rho$. The bound of bias can be obtained by bounding the local non-uniformity of $g$ in $B(\nu_i,\rho)$ if $\nu_i\leq \rho$. On the contrary, if $\nu_i>\rho$, we use assumption (d) to ensure that $\nu_i$ will not be too large, and thus will not cause significant estimation error. We let $\rho$ to decay with $M$ at a proper speed, to maximize the overall convergence rate of the bias. 
	
	To bound the bias in $S_2$, we let the threshold between $S_1$ and $S_2$ to decay with sample size $M$, so that the probability mass of $S_2$ also decreases with $M$. We then combine the bound of $S_1$ and $S_2$, and adjust the rate of the decay of the threshold between $S_1$ and $S_2$ properly.
	
	The detailed proof can be found in Appendix~\ref{sec:unbounded}.
\end{proof}

\section{Variance Analysis}\label{sec:variance}

We now discuss the variance of this divergence estimator, based on the following unifying assumptions.
\begin{assum}\label{ass:var}
	Assume that the following conditions hold:
	
	(a) $f$ and $g$ are continuous almost everywhere;
	
	(b) $\exists r_0>0$, such that
	\begin{eqnarray}
	\int f(\mathbf{x})\left(\underset{r<r_0}{\inf}\tilde{f}(\mathbf{x},r)\right)^2 d\mathbf{x}<\infty;\\
	\int f(\mathbf{x})\left(\underset{r<r_0}{\sup}\tilde{f}(\mathbf{x},r)\right)^2 d\mathbf{x}<\infty;\\	
	\int f(\mathbf{x})\left(\underset{r<r_0}{\inf}\tilde{g}(\mathbf{x},r)\right)^2 d\mathbf{x}<\infty;\\
	\int f(\mathbf{x})\left(\underset{r<r_0}{\sup}\tilde{g}(\mathbf{x},r)\right)^2 d\mathbf{x}<\infty,	
	\label{eq:varb}	
	\end{eqnarray}
	in which
	\begin{eqnarray}
	\tilde{f}(\mathbf{x},r)=P_f(B(\mathbf{x},r))/V(B(\mathbf{x},r))
	\label{eq:tildedef}
	\end{eqnarray} 
	is the average of $f$ over $B(\mathbf{x},r)$. $\tilde{g}$ is similarly defined;
	
	(c) $\mathbb{E}[\norm{\mathbf{X}}^s]\leq K$ and $\mathbb{E}[\norm{\mathbf{Y}}^s]\leq K$ for two finite constants $s,K>0$;
	
	(d) There exist two constants $C$ and $U_g$, such that for all $\mathbf{x}$, $f(\mathbf{x})\leq Cg(\mathbf{x})$ and $g(\mathbf{x})\leq U_g$.
	
\end{assum}

Assumption \ref{ass:var} (a)-(c) are satisfied if either Assumption \ref{ass:bounded} or Assumption \ref{ass:unbounded} is satisfied. (a) only requires that the pdf is continuous almost everywhere, and thus holds not only for distributions that are smooth everywhere, but also for distributions that have boundaries. (b) is obviously satisfied under Assumption \ref{ass:bounded}, since it requires that the densities are both upper and lower bounded. From Assumption \ref{ass:unbounded}, it is also straightforward to show that $\int f(\mathbf{x})\ln^2 f(\mathbf{x})d\mathbf{x}<\infty$ and $\int f(\mathbf{x})\ln^2 g(\mathbf{x})<\infty$. This property combining with the smoothness condition (Assumption \ref{ass:unbounded} (c)) imply that \eqref{eq:varb} holds for sufficiently small $r_0$. (c) is the same as Assumption \ref{ass:unbounded} (d) and weaker than Assumption \ref{ass:bounded} (d). Therefore, (a)-(c) are weaker than both previous assumptions on the analysis of bias. (d) is a new assumption which restricts the density ratio. This is important since if the density ratio can be too large, which means that there exists a region on which there are too many samples from $\{\mathbf{X}_1,\ldots,\mathbf{X}_N \}$, but much fewer samples from $\{\mathbf{Y}_1,\ldots,\mathbf{Y}_M \}$, then $\nu_i$ will be large and unstable for too many $i\in \{1,\ldots,N\}$. Therefore we use assumption (d) to bound the density ratio. 

Under these assumptions, the variance of the divergence estimator can be bounded using the following theorem.

\begin{thm}\label{thm:var}
	Under Assumption \ref{ass:var}, if $N\ln M/M\rightarrow \infty$, then the convergence rate of the variance of estimator~\eqref{eq:estimator} can be bounded by:
	\begin{eqnarray}
	\Var[\hat{D}(f||g)]=\mathcal{O}\left(\frac{1}{N}+\frac{\ln^4 M\ln^2 (M+N)}{M}\right).
	\label{eq:var}
	\end{eqnarray}
\end{thm}
\begin{proof}
	(Outline) From \eqref{eq:estimator}, we have
	\begin{eqnarray}
	\Var[\hat{D}(f||g)]	&=&\Var\left[\frac{d}{N}\sum_{i=1}^N \ln \nu_i-\frac{d}{N}\sum_{i=1}^N \ln \epsilon_i\right]\nonumber\\
	&\leq &2\Var\left[\frac{d}{N}\sum_{i=1}^N\ln \epsilon_i\right]+2\Var\left[\frac{d}{N}\sum_{i=1}^N\ln \nu_i\right].
	\label{eq:vardecomp}
	\end{eqnarray}
	Our proof uses some techniques from \cite{biau2015lectures}, which proved the $\mathcal{O}(1/N)$ convergence of variance of Kozachenko-Leonenko entropy estimator with $k=1$ for one dimensional distributions, and \cite{zhao2019analysis}, which generalizes the result to arbitrary fixed dimension and $k$, without restrictions on the boundedness of the support. The basic idea is that if one sample is replaced by another i.i.d sample, then it can be shown that the $k$-NN distance will change only for a tiny fraction of the samples. 
	
	The first term in \eqref{eq:vardecomp} is just the variance of Kozachenko-Leonenko entropy estimator. Therefore we can use similar proof procedure as was already used in the proof of Theorem 2 in \cite{zhao2019analysis}. \cite{zhao2019analysis} analyzed a truncated Kozachenko-Leonenko entropy estimator, which means that $\epsilon_i$ is truncated by an upper bound $a_N$. We prove the same convergence bound for the estimator without truncation.
	
	For the second term in \eqref{eq:vardecomp}, the analysis becomes much harder, since the $k$-NN distance may change for much more samples from $\{\mathbf{X}_1,\ldots,\mathbf{X}_N \}$, instead of only a tiny fraction of samples. For this term, we design a new method to obtain the high probability bound of the deviation of $(d/N)\sum_{i=1}^N \ln \nu_i$ from its mean. The basic idea of our new methods can be briefly stated as following: Define two sets $S_1$ and $S_1'$, in which $S_1$ is a subset of $\mathbb{R}^d$ such that for any $\mathbf{x}\in S_1$, $\mathbf{Y}_1$ is among the $k$ nearest neighbors of $\mathbf{x}$ in  $\{\mathbf{Y}_1,\ldots,\mathbf{Y}_M \}$. Similarly, define $S_1'$ to be a set such that for all $\mathbf{x}\in S_1'$, $\mathbf{Y}_1'$ is among the $k$ nearest neighbors of $\mathbf{x}$. If we replace $\mathbf{Y}_1$ with $\mathbf{Y}_1'$, the kNN distance of $\mathbf{X}_i, i=1,\ldots, N$ will only change if $\mathbf{X}_i\in S_1$ or $\mathbf{X}_i\in S_1'$. With this observation, we give a high probability bound of the number of samples from $\{\mathbf{X}_1,\ldots,\mathbf{X}_N \}$ that are in $S_1$ and $S_1'$ respectively, and then bound the maximum difference of the estimated result caused by replacing $\mathbf{Y}_1$ with $\mathbf{Y}_1'$. Based on this bound, we can then bound the second term in \eqref{eq:vardecomp} using Efron-Stein inequality. 
	
		The detailed proof can be found in Appendix~\ref{sec:var}.
\end{proof}

In the analysis above, we have derived the convergence rate of bias and variance. With these results, we can then bound the mean square error of kNN based KL divergence estimator. For distributions that satisfy Assumptions \ref{ass:bounded} and \ref{ass:var}, the mean square error can be bounded by
\begin{eqnarray}
\mathbb{E}[(\hat{D}(f||g)-D(f||g))^2]=\mathcal{O}\left(M^{-\frac{2}{d}}\ln^\frac{2}{d}M+M^{-1}\ln^4 M\ln^2(M+N)+N^{-\frac{2}{d}}\ln^\frac{2}{d}N+N^{-1}\right).
\label{eq:mse1}
\end{eqnarray}

For distributions that satisfy Assumptions \ref{ass:unbounded} and \ref{ass:var}, the corresponding bound is
\begin{eqnarray}
\mathbb{E}[(\hat{D}(f||g)-D(f||g))^2]=\mathcal{O}\left(M^{-\frac{4\gamma}{d+2}}\ln^2 M +M^{-1}\ln^4 M\ln^2(M+N)+N^{-\frac{4\gamma}{d+2}}\ln^2 N+N^{-1}\right).
\label{eq:mse2}
\end{eqnarray}

\section{Minimax Analysis}\label{sec:mmx}
In this section, we derive the minimax lower bound of the mean square error of KL divergence estimation, which holds for all methods (not necessarily kNN based) that do not have the knowledge of the distributions $f$ and $g$. The minimax analysis also considers two cases, i.e. the distributions whose densities are bounded away from zero, and those who has approaching zero densities.

For the first case, the following theorem holds.
\begin{thm}\label{thm:mmx1}
	Define $\mathcal{S}_a$ as set of pairs $(f,g)$ that satisfies Assumptions \ref{ass:bounded} and \ref{ass:var}, and
	\begin{eqnarray}
	R_a(N,M):=\underset{\hat{D}}{\inf}\underset{(f,g)\in \mathcal{S}_a}{\sup}\mathbb{E}[(\hat{D}(N,M)-D(f||g))^2],
	\label{eq:Ra}
	\end{eqnarray}
	in which $\hat{D}(N,M)$ is the estimation of KL divergence using $N$ samples drawn from distribution with pdf $f$ and $M$ samples from $g$. Then for sufficiently large $U_f$, $U_g$, $H_f$, $H_g$ and sufficiently small $L_f$ and $L_g$, we have
	\begin{eqnarray}
	R_a(N,M)&=&\Omega\left(\frac{1}{N}+N^{-\frac{2}{d}\left(1+\frac{2}{\ln\ln N}\right)}\ln^{-2}N \ln^{-\left(2-\frac{2}{d}\right)}(\ln N)\right.\nonumber\\
	&&\left.+\frac{1}{M}+M^{-\frac{2}{d}\left(1+\frac{2}{\ln\ln M}\right)}\ln^{-2}M \ln^{-\left(2-\frac{2}{d}\right)}(\ln M)\right).
	\label{eq:mmx1}
	\end{eqnarray}
\end{thm}
\begin{proof} (Outline)
	The minimax lower bound of functional estimation can be bounded using Le Cam's method \cite{tsybakov2009introduction}. For the proof of Theorem \ref{thm:mmx1}, we use some techniques from \cite{wu2016minimax}, which derived the minimax bound of entropy estimation for discrete distributions. The main idea is to construct a subset of distributions that satisfy Assumptions \ref{ass:bounded} and \ref{ass:var}, and then conduct Poisson sampling. These operations can help us calculate the distance between two distributions in a more convenient way, which is important for using Le Cam's method. Details of the proof can be found in Appendix~\ref{sec:mmx1}.
\end{proof}
\eqref{eq:mmx1} can be simplified as
\begin{eqnarray}
R_a(N,M)=\Omega\left(\frac{1}{N}+\frac{1}{M}+N^{-\left(\frac{2}{d}+\delta\right)}+M^{-\left(\frac{2}{d}+\delta\right)}\right),
\label{eq:mmx1new}
\end{eqnarray} 
for arbitrarily small $\delta>0$. 

We remark that in Theorem \ref{thm:mmx1}, the support set $S_f$ and $S_g$ of pdfs $f$ and $g$ are unknown. If we assume that $S_f$ and $S_g$ are known, then with some boundary correction methods, such as the mirror reflection method proposed in \cite{liu2012exponential}, the convergence rate can be faster than that in \eqref{eq:mmx1}. However, in Theorem \ref{thm:mmx1}, instead of using fixed support sets, $\mathcal{S}_a$ contains distributions with a broad range of different support sets. These support sets are only restricted by Assumption \ref{ass:bounded} (c) and (d), which require that the surface area of all the elements in $\mathcal{S}_a$ are bounded by $H_f$ and $H_g$, and the diameters are bounded by $R$. As a result, the minimax convergence rate becomes slower. This result indicates the inherent difficulty caused by the boundary effect for distributions with densities bounded away from zero.

For the second case, the corresponding result is shown in Theorem \ref{thm:mmx2}.
\begin{thm}\label{thm:mmx2}
	Define $\mathcal{S}_b$ as set of pairs $(f,g)$ that satisfies Assumptions \ref{ass:unbounded} and \ref{ass:var}, and
	\begin{eqnarray}
	R_b(N,M):=\underset{\hat{D}}{\inf}\underset{(f,g)\in \mathcal{S}_b}{\sup}\mathbb{E}[(\hat{D}(N,M)-D(f||g))^2],
	\label{eq:Rb}	
	\end{eqnarray}
	then for sufficiently large $\mu,C_0,K$,
	\begin{eqnarray}
	R_b(N,M)=\Omega\left(\frac{1}{M}+\frac{1}{N}+M^{-\frac{4\gamma}{d+2}}(\ln M)^{-\frac{4d+8-4\gamma}{d+2}}+N^{-\frac{4\gamma}{d+2}}(\ln N)^{-\frac{4d+8-4\gamma}{d+2}}\right).
	\label{eq:mmx2}
	\end{eqnarray}
\end{thm}

\begin{proof} (Outline)
	The minimax convergence rate of differential entropy estimation under similar assumptions was derived in \cite{zhao2019analysis}. We can extend the analysis to the minimax convergence rate of cross entropy estimation between $f$ and $g$. Combine the bound for entropy and cross entropy, we can then obtain the minimax lower bound of the mean square error of KL divergence estimation. The detailed proof is shown in Appendix \ref{sec:mmx2}.
\end{proof}

Comparing \eqref{eq:mmx1new} with \eqref{eq:mse1}, as well as \eqref{eq:mmx2} with \eqref{eq:mse2}, we observe that the convergence rate of the upper bound of mean square error of kNN based KL divergence estimator nearly matches the minimax lower bound for both cases. These results indicate that the kNN method with fixed $k$ is nearly minimax rate optimal. 

\section{Numerical Examples}\label{sec:numerical}
In this section, we provide numerical experiments to illustrate the theoretical results in this paper. In the simulation, we plot the curve of the estimated bias and variance over sample sizes. For illustration simplicity, we assume that the sample sizes for two distributions are equal, i.e. $M=N$. For each sample size, the bias and variance are estimated by repeating the simulation $T$ times, and then calculate the sample mean and the sample variance of all these trials. For low dimensional distributions, the bias is relatively small, therefore it is necessary to conduct more trials comparing with high dimensional distributions. In the following experiments, we repeat $T=100,000$ times if $d=1$, and $10,000$ times if $d>1$. In all of the figures, we use log-log plots with base $10$. In all of the trials, we fix $k=3$.


 Figure \ref{fig:A} shows the convergence rate of kNN based KL divergence estimator for two uniform distributions with different support. This case is an example that satisfies Assumption \ref{ass:bounded}. In Figure \ref{fig:B}, $f$ and $g$ are two Gaussian distributions with different mean but equal variance. In Figure \ref{fig:C}, $f$ and $g$ are two Gaussian distributions with the same mean but different variance. These two cases are examples that satisfy Assumption \ref{ass:unbounded}.

\begin{figure}[h!]
	\subfigure[Bias]{\includegraphics[width=0.49\linewidth]{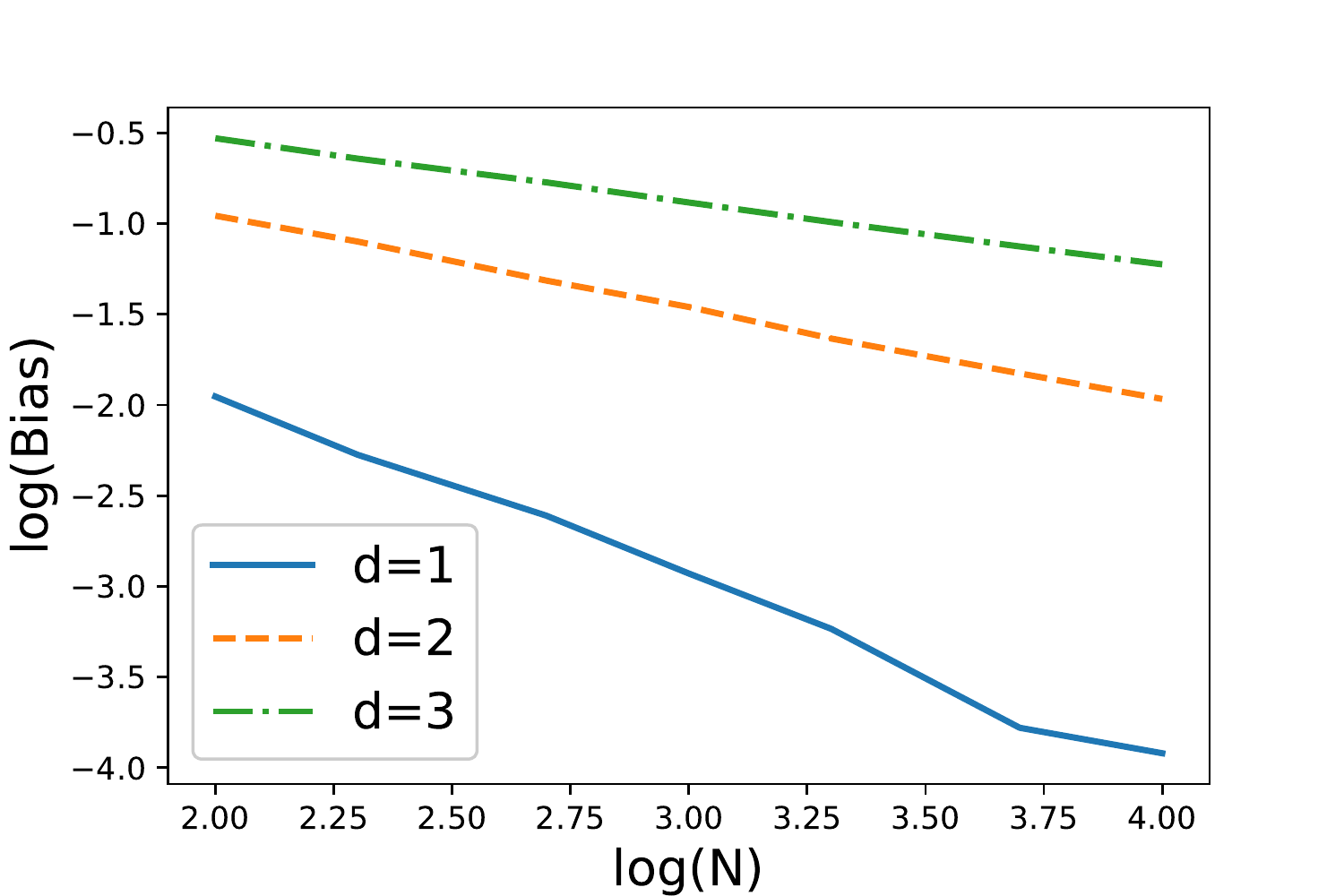}}	
	\subfigure[Variance]{\includegraphics[width=0.49\linewidth]{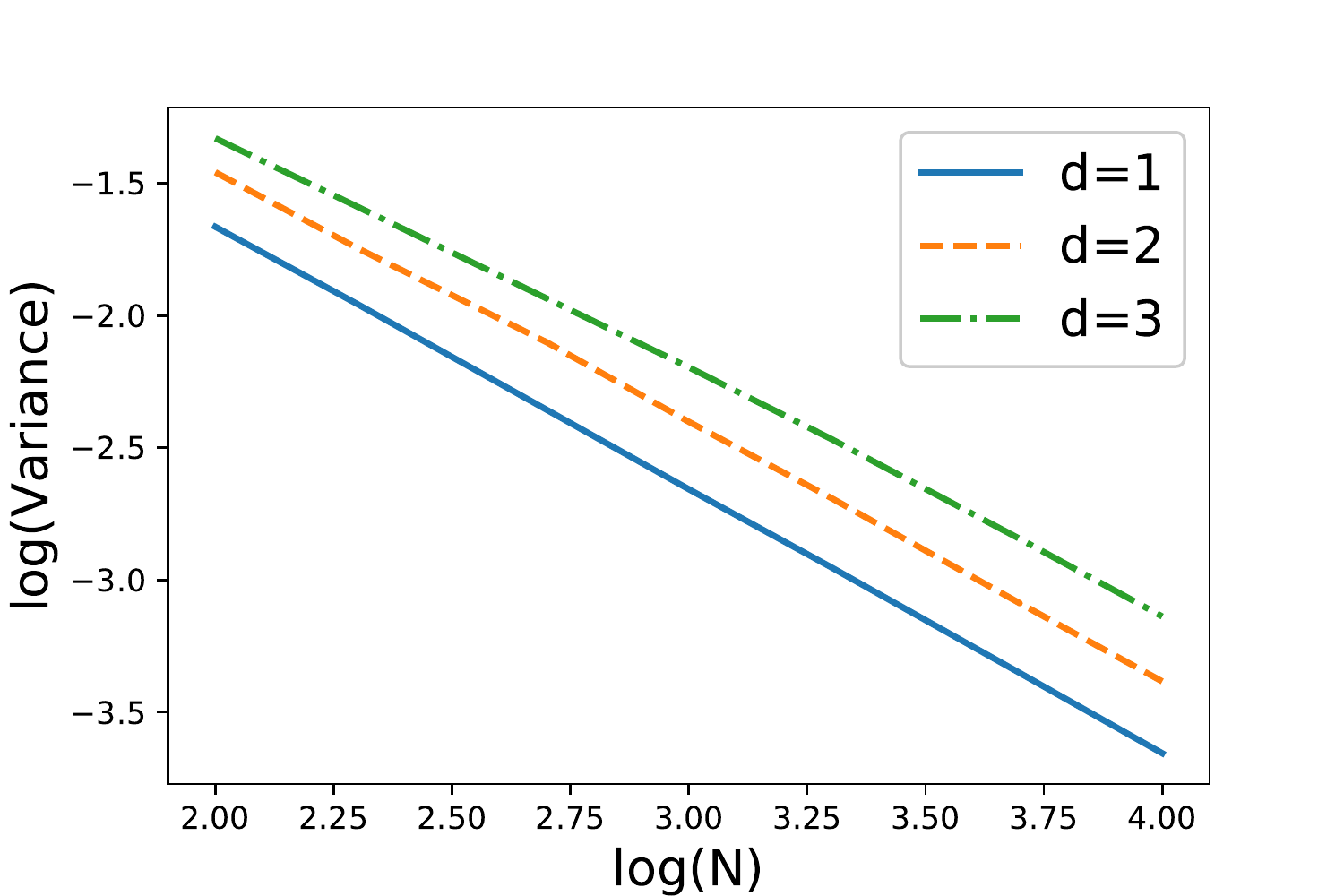}}							
	\caption{Convergence of bias and variance of kNN based KL divergence estimator for two uniform distributions with different support sets. $f=1$ in $[0.5,1.5]^d$, and $g=2^{-d}$ in $[0,2]^d$.}	\label{fig:A}
\end{figure}

\begin{figure}[h!]
	\subfigure[Bias]{\includegraphics[width=0.49\linewidth]{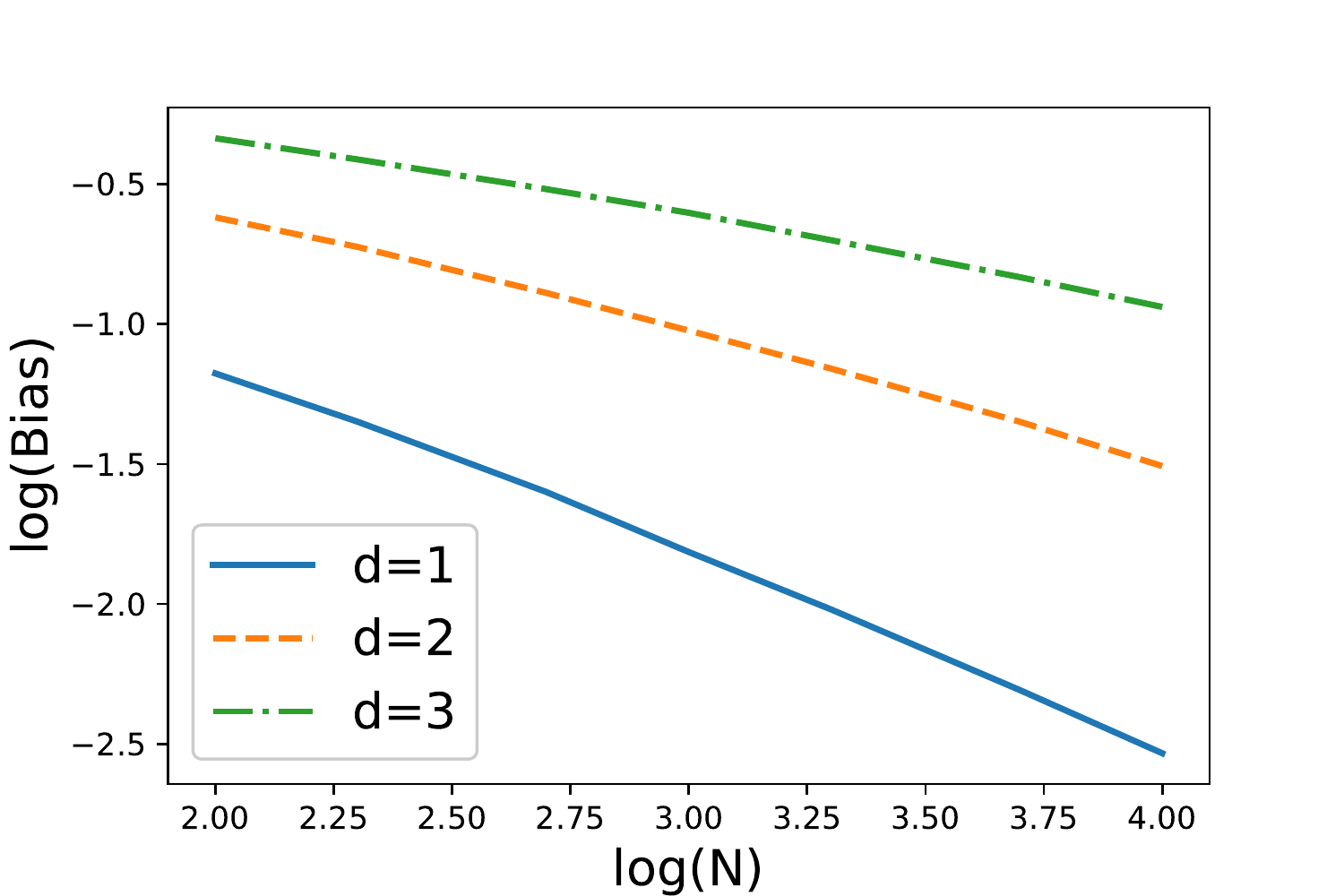}}	
	\subfigure[Variance]{\includegraphics[width=0.49\linewidth]{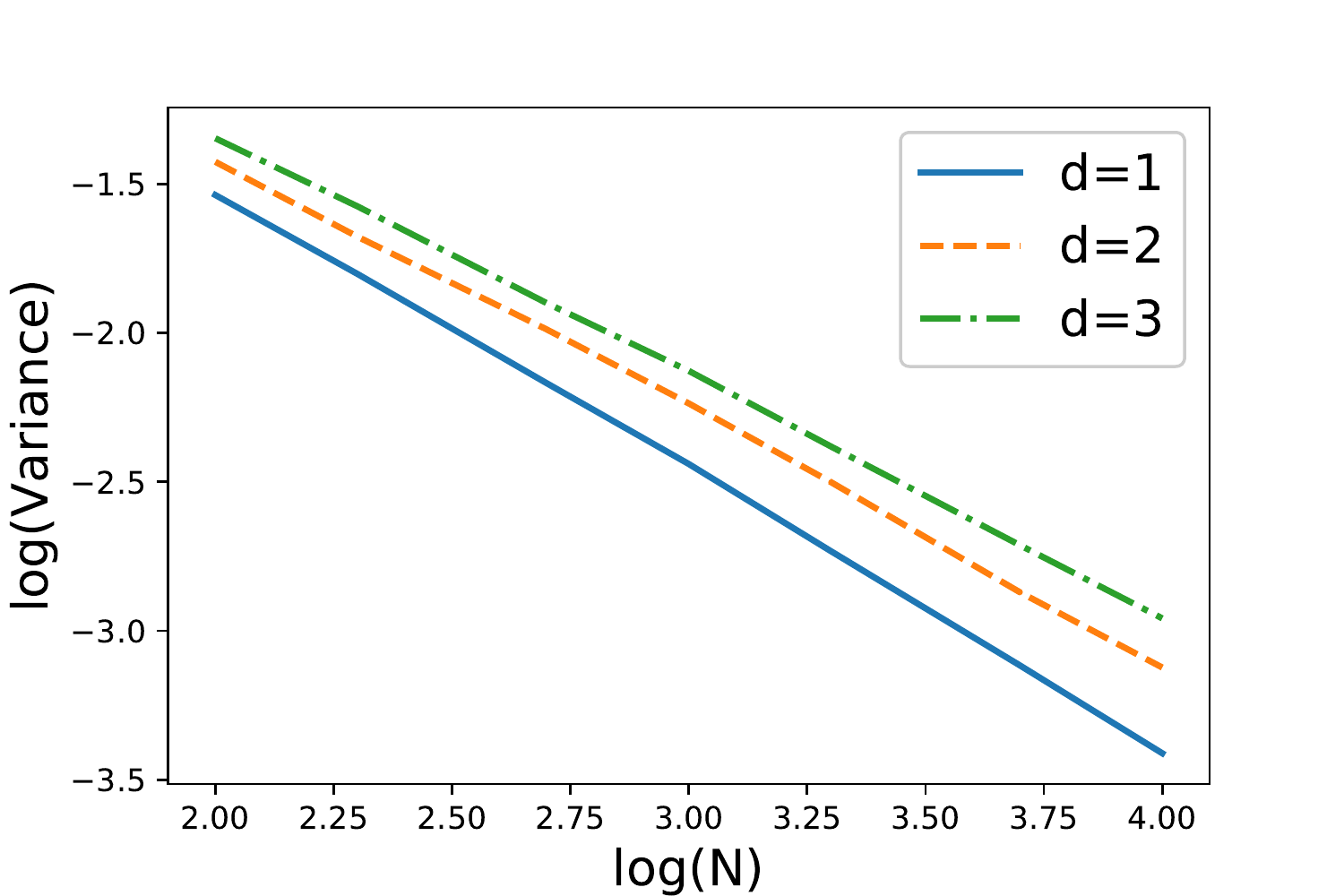}}							
	\caption{Convergence of bias and variance of kNN based KL divergence estimator for two Gaussian distributions with different means. $f$ is the pdf of $\mathcal{N}(\mathbf{0},\mathbf{I}_d)$, and $g$ is the pdf of $\mathcal{N}(\mathbf{1},\mathbf{I}_d)$, in which $\mathbf{I}_d$ denotes $d$ dimensional identity matrix, and $\mathbf{1}=(1,\ldots,1)$. }	\label{fig:B}
\end{figure}

\begin{figure}[h!]
	\subfigure[Bias]{\includegraphics[width=0.49\linewidth]{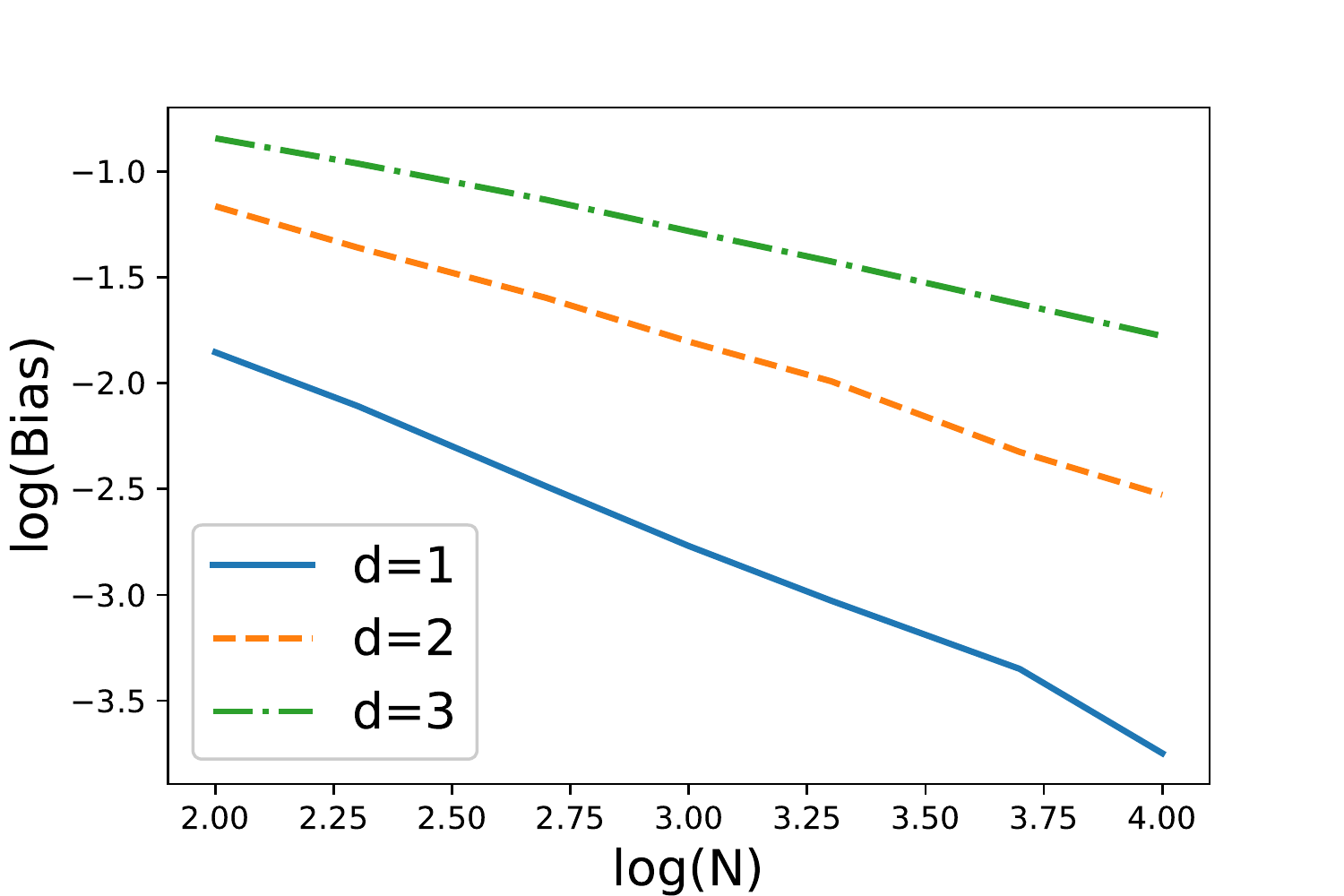}}	
	\subfigure[Variance]{\includegraphics[width=0.49\linewidth]{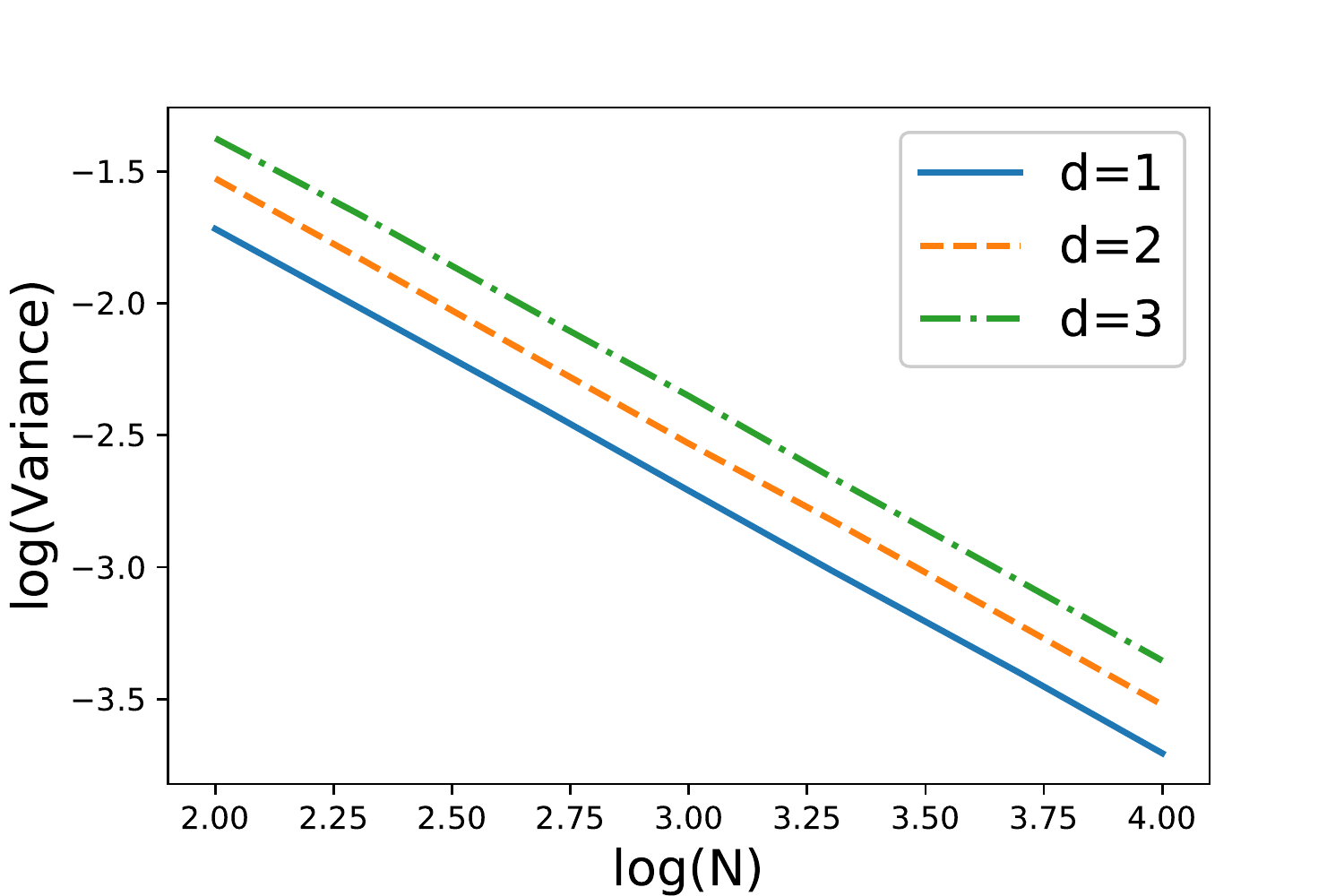}}							
	\caption{Convergence of bias and variance of kNN based KL divergence estimator for two Gaussian distributions with different variances. $f$ is the pdf of $\mathcal{N}(\mathbf{0},\mathbf{I}_d)$, and $g$ is the pdf of $\mathcal{N}(\mathbf{0},2\mathbf{I}_d)$.}	\label{fig:C}
\end{figure}

For all of these distributions above, we compare the empirical convergence rates of the bias and variance with the theoretical prediction. The empirical convergence rates are calculated by finding the negative slope of the curves in these figures by linear regression, while the theoretical ones come from Theorems \ref{thm:bounded}, \ref{thm:unbounded} and \ref{thm:var} respectively. The results are shown in Table \ref{tab:compare}. For the convenience of expression, we say that the theoretical convergence rate of bias or variance is $\beta$, if it decays with either $O(N^{-\beta})$ or $O(N^{-\beta+\delta})$ for arbitrarily small $\delta>0$, given the condition $M=N$.
  
\begin{table}[h]
	\caption{Theoretical and empirical convergence rate comparison}
	\label{tab:compare}
	\vskip 0.15in
	\begin{center}
		\begin{small}
			
			\begin{tabular}{|c|c|c|c|c|c|c|}
				\hline
				&\multicolumn{3}{|c|}{Bias, Empirical/Theoretical}  & \multicolumn{3}{|c|}{Variance, Empirical/Theoretical} \\
				\hline
				& $d=1$&$d=2$&$d=3$ &$d=1$&$d=2$&$d=3$\\
				\hline
				Fig.\ref{fig:A}  & 1.01/1.00&0.51/0.50 &0.34/0.33 &1.00/1.00&0.98/1.00&0.96/1.00 \\
				Fig.\ref{fig:B}&0.68/0.67 &0.47/0.50&0.36/0.40&0.94/--&0.85/--&0.81/--\\
				Fig.\ref{fig:C}&0.90/0.67 &0.68/0.50&0.45/0.40 &0.99/1.00&1.00/1.00&0.99/1.00 \\
				\hline
			\end{tabular}
			
		\end{small}
	\end{center}
	\vskip -0.1in
\end{table}

In Table \ref{tab:compare}, we observe that for the distribution used in Figure \ref{fig:A}, the empirical convergence rates of both bias and variance agree well with the theoretical prediction, in which the theoretical bound of bias comes from Theorem \ref{thm:bounded}, while the variance comes from Theorem \ref{thm:var}. 

For the distribution in Figure \ref{fig:B}, the empirical convergence of bias matches the theoretical prediction from Theorem \ref{thm:unbounded}. For Gaussian distributions with different mean, it can be shown that for any $\gamma<1$, there exists a constant $\mu$ such that Assumption \ref{ass:unbounded} (b) holds. Therefore, according to Theorem \ref{thm:unbounded}, the convergence rate of bias is $\mathcal{O}(N^{-\frac{2}{d+2}+\delta})$ for arbitrarily small $\delta>0$. Therefore, in the second line of Table \ref{tab:compare}, the theoretical rate of bias is $0.67$, $0.50$ and $0.40$, respectively. Now we discuss the convergence rate of variance. Note that the theoretical result about the variance is unknown, since $f/g$ can reach infinity, thus Assumption \ref{ass:var} (d) is not satisfied, and Theorem \ref{thm:var} does not hold here. We observe that the empirical convergence rate is slower than that in other cases. Such a result may indicate that it is harder to estimate the KL divergence if the density ratio is unbounded.

For the distribution in Figure \ref{fig:C}, the empirical and theoretical convergence rate of the variance matches well, while the empirical rate of bias is faster than the theoretical prediction. Note that the bound we have derived holds universally for all distributions that satisfy the assumptions. For certain specific distribution, the convergence rate can probably be faster. In particular, there is an uniform bound on the Hessian of $f$ and $g$ in Assumption \ref{ass:unbounded} (c). However, for Gaussian distributions, the Hessian is lower where the pdf value is small. Therefore, the local non-uniformity is not as serious as the worst case that satisfies the assumptions. 

\section{Conclusion}\label{sec:conclusion}

In this paper, we have analyzed the convergence rates of the bias and variance of the kNN based KL divergence estimator proposed in \cite{wang2009divergence}. For the bias, we have discussed two types of distributions depending on the main causes of the bias. In the first case, the distribution has bounded support, and the pdf is bounded away from zero. In the second case, the distribution is smooth everywhere and the pdf can approach zero arbitrarily close. For the variance, we have derived the convergence rate under a more general assumption. Furthermore, we have derived the minimax lower bound of KL divergence estimation. The bound holds for all possible estimators. We have shown that for both types of distributions, the kNN based KL divergence estimator is nearly minimax rate optimal. We have also used numerical experiments to illustrate that the practical performances of kNN based KL divergence estimator are consistent with our theoretical analysis.




\appendices
\section{Proof of Theorem \ref{thm:bounded}}\label{sec:bounded}
According to \eqref{eq:estimator},
\begin{eqnarray}
\mathbb{E}[\hat{D}(f||g)]-D(f||g)&=&\frac{d}{N}\mathbb{E}[\ln \nu-\ln \epsilon]+\ln \frac{M}{N-1}-\mathbb{E}[\ln f(\mathbf{X})]+\mathbb{E}[\ln g(\mathbf{X})]\nonumber\\
&=&-\left[-\psi(k)+\psi(N)+\ln c_d+d\mathbb{E}[\ln \epsilon]+\mathbb{E}[\ln f(\mathbf{X})]\right]\nonumber\\
&&+\left[-\psi(k)+\psi(M+1)+\ln c_d+d\mathbb{E}[\ln \nu] +\mathbb{E}[\ln g(\mathbf{X})]\right]\nonumber\\
&&+\ln M-\psi(M+1)-\ln(N-1)+\psi(N)\nonumber\\
&:=&-I_1+I_2+I_3,
\end{eqnarray}
in which
\begin{eqnarray}
	I_1&=&-\psi(k)+\psi(N)+\ln c_d+d\mathbb{E}[\ln \epsilon]+\mathbb{E}[\ln f(\mathbf{X})],\\
	I_2&=&-\psi(k)+\psi(M+1)+\ln c_d+d\mathbb{E}[\ln \nu] +\mathbb{E}[\ln g(\mathbf{X})],\\
	I_3&=&\ln M-\psi(M+1)-\ln(N-1)+\psi(N),
\end{eqnarray} and $c_d$ is the volume of unit ball. Here, we omit $i$, since $\mathbb{E}[\ln \epsilon(i)]$ and $\mathbb{E}[\ln \nu(i)]$ are the same for all $i$.

In the following, we provide details on how to bound $I_2$. $I_1$ can then be bounded using similar method.

To begin with, we denote $P_g(S)$ as the probability mass of $S$ under pdf $g$, i.e. $P_g(S)=\int_S g(\mathbf{x})d\mathbf{x}$. We have the following lemma.
\begin{lem}\label{lem:pdf}
	There exists a constant $C_1$, such that, if $B(\mathbf{x},r)\subset S_g$, we have $$|P_g(B(\mathbf{x},r))-c_dr^dg(\mathbf{x})|\leq C_1 r^{d+2}.$$
\end{lem}
\begin{proof}
	\begin{eqnarray}
	|P_g(B(\mathbf{x},r))-g(\mathbf{x})c_dr^d|&=&\left|\int_{B(\mathbf{x},r)}(g(\mathbf{u})-g(\mathbf{x}))d\mathbf{u}\right|\nonumber\\
	&\leq &\left|\int_{B(\mathbf{x},r)}\nabla g(\mathbf{x})(\mathbf{u}-\mathbf{x})d\mathbf{u}+\int_{B(\mathbf{x},r)}C_0(\mathbf{u}-\mathbf{x})^T(\mathbf{u}-\mathbf{x})d\mathbf{u}\right|\nonumber\\
	&\leq & C_0r^2V(B(\mathbf{x},r))\nonumber\\
	&=&C_0c_dr^{d+2},
	\end{eqnarray}
	in which the first inequality uses Assumption \ref{ass:bounded} (f).
\end{proof}

From order statistics \cite{david1970order}, $\mathbb{E}[\ln P_g(B(\mathbf{x},r))]=\psi(k)-\psi(M+1)$, therefore
\begin{eqnarray}
I_2=-\mathbb{E}\left[\ln \frac{P_g(B(\mathbf{X},\nu))}{c_d\nu^d g(\mathbf{X})}\right].
\label{eq:I2}
\end{eqnarray}

Define
\begin{eqnarray}
S_1&=&\{\mathbf{x}|B(\mathbf{x},a_M)\subset S_g \},\\
S_2&=&S_g\setminus S_1,
\end{eqnarray}
in which $a_M=A(\ln M/M)^{1/d}$, and $A=(2/(L_gc_d))^{1/d}$. From \eqref{eq:I2}, we observe that the bias is determined by the difference between the average pdf in $B(\mathbf{x},\nu)$ and the pdf at its center $g(\mathbf{x})$. $S_1$ is the region that is relatively far from the boundary. For all $\mathbf{x}\in S_1$, with high probability, $B(\mathbf{x},\nu)\subset S_g$. In this case, the bias is caused by the non-uniformity of density. With the increase of sample size, the effect of such non-uniformity will converge to zero. $S_2$ is the region near to the boundary, in which the probability that $B(\mathbf{x},\nu)\not\subset S$ is not negligible, hence $P(B(\mathbf{x},\nu))$ can deviate significantly comparing with $c_d\nu^d g(\mathbf{x})$. Therefore, the bias in this region will not converge to zero. However, we let the size of $S_2$ converge to zero, so that the overall bound of the bias converges.

For sufficiently large $M$,
\begin{eqnarray}
&&\left|\mathbb{E}\left[\left(\ln \frac{P_g(B(\mathbf{X},\nu))}{c_d\nu^d g(\mathbf{X})}\right)\mathbf{1}(\mathbf{X}\in S_1)\right]\right|\nonumber\\
&\leq& \left|\mathbb{E}\left[\left(\ln \frac{P_g(B(\mathbf{X},\nu))}{c_d\nu^d g(\mathbf{X})}\right)\mathbf{1}(\mathbf{X}\in S_1,\nu\leq a_M)\right]\right|
+\left|\mathbb{E}\left[\left(\ln \frac{P_g(B(\mathbf{X},\nu))}{c_d\nu^d g(\mathbf{X})}\right)\mathbf{1}(\mathbf{X}\in S_1,\nu > a_M)\right]\right|\nonumber\\
&\overset{(a)}{\leq}&\left|\mathbb{E}\left[\ln\left(1-\frac{C_1\nu^2}{c_dg(\mathbf{X})}\right)\mathbf{1}(\nu\leq a_M, \mathbf{X}\in S_1)\right]\right|+\ln \frac{U_g}{a L_g} \text{P}(\mathbf{X}\in S_1, \nu>a_M)\nonumber\\
&\overset{(b)}{\leq}&\frac{2C_1}{c_dL_g}a_M^2 +\ln \frac{U_g}{a L_g}\left(\frac{e}{k}\right)^k \frac{(2\ln M)^k}{M^2}\sim \left(\frac{\ln M}{M}\right)^\frac{2}{d}.
\label{eq:S1a}
\end{eqnarray}
In step (a), we use Lemma \ref{lem:pdf}, Assumption \ref{ass:bounded} (b) and Assumption \ref{ass:bounded} (e). In step (b), the first term uses the fact that for sufficiently large $M$, $a_M$ will be sufficiently small, hence $C_1\nu^2/(c_dg(\mathbf{x}))\leq C_1a_M^2/(c_dg(\mathbf{x}))<1/2$. The second term of step (b) comes from the Chernoff bound, which indicates that for all $\mathbf{x}\in S_1$ and sufficiently large $M$,
\begin{eqnarray}
\text{P}(\nu>a_M|\mathbf{x})&\leq& e^{-MP_g(B(\mathbf{x},a_M))}\left(\frac{eMP_g(B(\mathbf{x},a_M))}{k}\right)^k\nonumber\\
&\leq & e^{-ML_gc_da_M^d}\left(\frac{eML_gc_da_M^d}{k}\right)^k\nonumber\\
&=&\left(\frac{e}{k}\right)^k \frac{(2\ln M)^k}{M^2}.
\end{eqnarray}

Moreover, 
\begin{eqnarray}
\left|\mathbb{E}\left[\ln \frac{P_g(B(\mathbf{X},\nu))}{c_d\nu^d g(\mathbf{X})} \mathbf{1}(\mathbf{X}\in S_2)\right]\right|&\leq & \ln \frac{U_g}{a L_g} \text{P}(\mathbf{X}\in S_2)\nonumber\\
&\leq & \ln \frac{U_g}{a L_g} U_g V(S_2)\nonumber\\
&\leq & \ln \frac{U_g}{a L_g} U_g H_g a_M\nonumber\\
&\sim & \left(\frac{\ln M}{M}\right)^\frac{1}{d}.
\label{eq:S2a}
\end{eqnarray}
In this equation, $V(S_2)$ is the volume of $S_2$, and we use the fact that $V(S_2)\leq H_ga_M$ according to the definition of $S_2$ and Assumption \ref{ass:bounded} (c). Based on \eqref{eq:S1a} and \eqref{eq:S2a}, 
\begin{eqnarray}
|I_2|\lesssim \left(\frac{\ln M}{M}\right)^\frac{1}{d}.
\end{eqnarray}
Similarly, we have $|I_1|\lesssim (\ln N/N)^{(1/d)}$, and according to the definition of digamma function $\psi$, $|I_3|\lesssim 1/M+1/N$. Therefore
\begin{eqnarray}
|\mathbb{E}[\hat{D}(f||g)]-D(f||g)|\lesssim \left(\frac{\ln \min\{M,N \}}{\min\{M,N\}}\right)^\frac{1}{d}.
\end{eqnarray}
\section{Proof of Theorem \ref{thm:unbounded}}\label{sec:unbounded}
In this section, we derive the bound of the bias for distributions that satisfy Assumption \ref{ass:unbounded}. These distributions are smooth everywhere and the densities can approach zero. We begin with the following lemmas, whose proofs can be found in Appendix \ref{sec:ub}, \ref{sec:log}, and \ref{sec:tail}, respectively.
\begin{lem}\label{lem:ub}
	There exist constants $U_f$ and $U_g$ such that $f(\mathbf{x})\leq U_f$ and $g(\mathbf{x})\leq U_g$ for all $\mathbf{x}$.
\end{lem}
\begin{lem}\label{lem:log}
	There exists a constant $C_2$, such that $$\mathbb{E}[|\ln \norm{\mathbf{X}}|\mathbf{1}(g(\mathbf{X})\leq t)]\leq C_2t^\gamma \ln(1/t)$$ for sufficiently small $t$, in which $\mathbf{X}$ follows a distribution with pdf $f$.
\end{lem}
\begin{lem}\label{lem:tail}
	For sufficiently small $t$,
	\begin{eqnarray}
	\int_{g(\mathbf{x})>t}\frac{f(\mathbf{x})}{g(\mathbf{x})}d\mathbf{x}\leq \left\{
	\begin{array}{ccc}
	\mu \left(1+\ln \frac{1}{\mu t}\right) &\text{if} & \gamma=1\\
	\frac{\mu}{1-\gamma} t^{\gamma-1}&\text{if} & \gamma<1.
	\end{array}
	\right.
	\end{eqnarray}
\end{lem}
Similar to the proof of Theorem \ref{thm:bounded}, we decompose the bias as $\mathbb{E}[\hat{D}(f||g)]-D(f||g)=-I_1+I_2+I_3$. Then
\begin{eqnarray}
|I_2|=\left|\mathbb{E}\left[\ln \frac{P_g(B(\mathbf{X},\nu))}{c_d\nu^dg(\mathbf{X})}\right]\right|.
\end{eqnarray}
Divide $S_g$ into two parts.
\begin{eqnarray}
S_1&=&\left\{\mathbf{x}|g(\mathbf{x})>\frac{2C_1}{c_d}a_M^2 \right\},\label{eq:S1def}\\
S_2&=&S_g\setminus S_1,
\end{eqnarray}
in which $a_M=AM^{-\beta}$, $A=(k/C_1)^{(1/(d+2))}$. $\beta$ will be determined later. $C_1$ is the constant in Lemma \ref{lem:pdf}.

We first consider the region $S_1$. 
\begin{eqnarray}
\left|\mathbb{E}\left[\ln \frac{P_g(B(\mathbf{x},\nu))}{c_d\nu^d g(\mathbf{X})} \mathbf{1}(\mathbf{X}\in S_1,\nu\leq a_M)\right]\right|&\overset{(a)}{\leq}&\left|\mathbb{E}\left[\ln\left(1-\frac{C_1a_M^2}{c_dg(\mathbf{X})}\right)\mathbf{1}(\mathbf{X}\in S_1,\nu\leq a_M)\right]\right|\nonumber\\
&\overset{(b)}{\leq}&\left|\mathbb{E}\left[\frac{2C_1a_M^2}{c_dg(\mathbf{X})}\mathbf{1}(\mathbf{X}\in S_1)\right]\right|\nonumber\\
&\lesssim& a_M^2\int_{g(\mathbf{x})>\frac{2C_1}{c_d}a_M^2} \frac{f(\mathbf{x})}{g(\mathbf{x})}d\mathbf{x}\nonumber\\
&\overset{(c)}{\lesssim} &\left\{
\begin{array}{ccc}
M^{-2\beta \gamma} &\text{if} &\gamma<1 \\
M^{-2\beta}\ln M &\text{if} & \gamma=1,
\end{array}
\right.
\label{eq:I21}
\end{eqnarray}
in which (a) comes from Lemma \ref{lem:pdf}. For (b), note that according to \eqref{eq:S1def}, $C_1a_M^2/(c_dg(\mathbf{x}))<1/2$ for $\mathbf{x}\in S_1$, and $|\ln (1-u)|\leq 2u$ for any $0<u\leq 1/2$. (c) uses Lemma \ref{lem:tail}. 

For $\nu>a_M$, note that according to Lemma \ref{lem:pdf},
\begin{eqnarray}
P_g(B(\mathbf{x},a_M))\geq c_da_M^dg(\mathbf{x})-C_1a_M^{d+2}\geq \frac{1}{2}c_da_M^dg(\mathbf{x}).
\label{eq:masslb}
\end{eqnarray} 

Based on this fact, if $\beta\leq 1/(d+2)$, we show the following two lemmas:
\begin{lem}\label{lem:largeradius}
	There exists a constant $C_3$, such that
	\begin{eqnarray}
	\text{P}(\nu>a_M,\mathbf{X}\in S_1)\leq C_3M^{-\gamma(1-\beta d)}.
	\end{eqnarray}
\end{lem}
\begin{proof}
	Please see Appendix \ref{sec:largeradius} for detailed proof.
\end{proof}
\begin{lem}\label{lem:lognu}
	There exists a constant $C_4$, such that
	\begin{eqnarray}
	\mathbb{E}\left[\ln \frac{\nu}{a_M}\mathbf{1}(\nu>a_M,\mathbf{X}\in S_1)\right]\leq C_4 M^{-\gamma(1-\beta d)}\ln M.
	\end{eqnarray}
\end{lem}
\begin{proof}
	Please see Appendix \ref{sec:lognu} for detailed proof.
\end{proof}
Then
\begin{eqnarray}
&&\left|\mathbb{E}\left[\ln \frac{P_g(B(\mathbf{X},\nu))}{c_d\nu^d g(\mathbf{X})}\mathbf{1}(\mathbf{X}\in S_1,\nu>a_M)\right]\right|\nonumber\\
&\leq& |\mathbb{E}[\ln P_g(B(\mathbf{X},a_M))\mathbf{1}(\mathbf{X}\in S_1,\nu>a_M)]|+|\mathbb{E}[\ln(c_da_M^d)\mathbf{1}(\mathbf{X}\in S_1,\nu>a_M)]|\nonumber\\
&&\hspace{1cm}+|\mathbb{E}[\ln g(\mathbf{X})\mathbf{1}(\mathbf{X}\in S_1,\nu>a_M)]|+d\left|\mathbb{E}\left[\ln \frac{\nu}{a_M}\mathbf{1}(\nu>a_M,\mathbf{X}\in S_1)\right]\right|.
\label{eq:largenu}
\end{eqnarray}
Note that
\begin{eqnarray}
1\geq P_g(B(\mathbf{x},a_M))\geq c_da_M^d g(\mathbf{x})-C_1a_M^{d+2}\geq C_1a_M^{d+2}=C_1A^{d+2}M^{-\beta(d+2)}.
\end{eqnarray}
Therefore 
\begin{eqnarray}
|\mathbb{E}[\ln P_g(B(\mathbf{X},a_M))\mathbf{1}(\mathbf{X}\in S_1, \nu>a_M)]|\lesssim M^{-\gamma(1-\beta d)}\ln M.
\end{eqnarray}
The second and the third terms in \eqref{eq:largenu} satisfy the same bound. The last term can be bounded using Lemma \ref{lem:lognu}. Hence
\begin{eqnarray}
\left|\mathbb{E}\left[\ln \frac{P_g(B(\mathbf{X},\nu))}{c_d\nu^d g(\mathbf{X})}\mathbf{1}(\mathbf{X}\in S_1,\nu>a_M)\right]\right|\lesssim M^{-\gamma(1-\beta d)}\ln M.
\label{eq:I22}
\end{eqnarray}
Now we consider $\mathbf{x}\in S_2$.
\begin{eqnarray}
\left|\mathbb{E}\left[\ln \frac{P_g(B(\mathbf{X},\nu))}{c_d\nu^d g(\mathbf{X})}\mathbf{1}(\mathbf{X}\in S_2)\right]\right|&\leq& |\mathbb{E}[\ln P_g(B(\mathbf{X},\nu))\mathbf{1}(\mathbf{X}\in S_2)]|+|\mathbb{E}[\ln g(\mathbf{X})\mathbf{1}(\mathbf{X}\in S_2)]|\nonumber\\
&&+|\ln c_d|\text{P}(\mathbf{X}\in S_2)+d|\mathbb{E}[\ln \nu \mathbf{1}(\mathbf{X}\in S_2)]|.
\label{eq:S2}
\end{eqnarray}
From order statistics \cite{david1970order}, $|\mathbb{E}[\ln P_g(B(\mathbf{x},\nu))|\mathbf{x}]|=|\psi(k)-\psi(M)|\leq \ln M$. According to Assumption~\ref{ass:unbounded} (b), the first three terms in \eqref{eq:S2} can be bounded by:
\begin{eqnarray}
|\mathbb{E}[\ln P_g(B(\mathbf{X},r))\mathbf{1}(\mathbf{X}\in S_2)]|\lesssim \ln M \text{P}(\mathbf{X}\in S_2)\sim \ln M a_M^{2\gamma}\sim M^{-2\beta \gamma}\ln M,
\end{eqnarray}
\begin{eqnarray}
|\mathbb{E}[\ln g(\mathbf{X})\mathbf{1}(\mathbf{X}\in S_2)]&=&\mathbb{E}\left[\ln \frac{1}{g(\mathbf{X})}\mathbf{1}\left(g(\mathbf{X})\leq \frac{2C_1}{c_d}a_M^2\right)\right]\nonumber\\
&=&\int_0^\infty \text{P}\left(\ln \frac{1}{g(\mathbf{X})}\mathbf{1}\left(g(\mathbf{X})\leq \frac{2C_1}{c_d}a_M^2\right)>t\right)dt\nonumber\\
&\leq & \int_0^{\ln \frac{c_d}{2C_1a_M^2}} \text{P}\left(g(\mathbf{X})\leq \frac{2C_1}{c_d}a_M^2\right)dt+\int_{\ln \frac{c_d}{2C_1a_M^2}}^\infty \text{P}\left(g(\mathbf{X})<e^{-t}\right)dt\nonumber\\
&\leq &\mu \left(\frac{2C_1}{c_d}a_M^2\right)^\gamma\ln \frac{c_d}{2C_1a_M^2}+\int_{\ln \frac{c_d}{2C_1a_M^2}}^\infty \mu e^{-\gamma t}dt\nonumber\\
&=&\mu \left(\frac{2C_1}{c_d}a_M^2\right)^\gamma\left(\ln \frac{c_d}{2C_1a_M^2}+\frac{1}{\gamma}\right)\nonumber\\
&\sim & M^{-2\beta \gamma}\ln M,
\end{eqnarray}
and
\begin{eqnarray}
|\ln c_d|\text{P}(\mathbf{X}\in S_2)|\lesssim M^{-2\beta \gamma}.
\end{eqnarray}

The last term in \eqref{eq:S2} can be bounded using the following lemma, whose proof can be found in Appendix~\ref{sec:log2}.
\begin{lem}\label{lem:log2}
	There exist two constants $C_5$ and $C_6$, such that for sufficiently large $M$,
	\begin{eqnarray}
	|\mathbb{E}[\ln \nu |\mathbf{x}]|\leq C_5\ln M+C_6|\ln \norm{\mathbf{x}}|.
	\end{eqnarray}
\end{lem}
Using this lemma, we have 
\begin{eqnarray}
|\mathbb{E}[\ln \nu \mathbf{1}(\mathbf{X}\in S_2)]\leq |\mathbb{E}[(C_5\ln M+C_6|\ln \norm{\mathbf{X}}|)\mathbf{1}(\mathbf{X}\in S_2)]|\lesssim a_M^{2\gamma}\ln \frac{1}{a_M}\sim M^{-2\beta \gamma}\ln M.
\end{eqnarray}

Therefore
\begin{eqnarray}
\left|\mathbb{E}\left[\ln \frac{P_g(B(\mathbf{X},\nu))}{c_d\nu^d g(\mathbf{X})} \mathbf{1}(\mathbf{X}\in S_2)\right]\right|\lesssim M^{-2\beta \gamma}\ln M.
\label{eq:I23}
\end{eqnarray}

Combining \eqref{eq:I21}, \eqref{eq:I22} and \eqref{eq:I23}, we get
\begin{eqnarray}
|I_2|\lesssim M^{-2\beta \gamma}\ln M+M^{-\gamma(1-\beta d)}\ln M.
\end{eqnarray}

Since the above bound holds for arbitrary $\beta\leq 1/(d+2)$, we just let $\beta=1/(d+2)$, then 
\begin{eqnarray}
|I_2|\lesssim M^{-\frac{2\gamma}{d+2}}\ln M.
\end{eqnarray}
Similarly, we have $|I_1|\lesssim N^{-\frac{2\gamma}{d+2}}\ln N$, and according to the definition of digamma function, $|I_3|\lesssim 1/M+1/N$. Hence
\begin{eqnarray}
|\mathbb{E}[\hat{D}(f||g)]-D(f||g)|\lesssim \left(\min\{M,N\}\right)^{-\frac{2\gamma}{d+2}}\ln \min\{M,N\}.
\end{eqnarray}

\subsection{Proof of Lemma \ref{lem:ub}}\label{sec:ub}
We only show that there exists a constant $U_g$ such that $g(\mathbf{x})\leq U_g$ holds for all $\mathbf{x}$. The proof of the upper bound $U_f$ of density $f$ will be exactly the same. From Lemma \ref{lem:pdf},
\begin{eqnarray}
P_g(B(\mathbf{x},r))\geq g(\mathbf{x})c_dr^d-C_1r^{d+2}.
\end{eqnarray}
Since $P_g(B(\mathbf{x},r))\leq 1$, we have 
\begin{eqnarray}
g(\mathbf{x})\leq \frac{1+C_1r^{d+2}}{c_dr^d}
\label{eq:gub}
\end{eqnarray}
for all $r>0$. Define $U_g$ as the right hand side of \eqref{eq:gub} given $r=(d/(2C_1))^{1/(d+2)}$, i.e.
\begin{eqnarray}
U_g=\frac{1+\frac{d}{2}}{c_d\left(\frac{d}{2C_1}\right)^\frac{d}{d+2}},
\end{eqnarray} 
then $g(\mathbf{x})\leq U_g$ for all $\mathbf{x}$.

\subsection{Proof of Lemma \ref{lem:log}}\label{sec:log}
From H{\"o}lder inequality, For any $p$, $q$ such that $p>1$, $q>1$, and $1/p+1/q=1$,
\begin{eqnarray}
\mathbb{E}\left[\ln \norm{\mathbf{x}}|\mathbf{1}(g(\mathbf{X})\leq t)\right]\leq \left(\mathbb{E}\left[|\ln \norm{\mathbf{x}}|^p\right]\right)^\frac{1}{p}\left(\mathbb{E}\left[\mathbf{1}(g(\mathbf{X})\leq t)^q\right]\right)^\frac{1}{q}.
\end{eqnarray}
From Assumption \ref{ass:unbounded} (b),
\begin{eqnarray}
\mathbb{E}[\mathbf{1}(g(\mathbf{X})\leq t)^q]=\text{P}(g(\mathbf{X})\leq t)\leq \mu t^\gamma.
\end{eqnarray}
Moreover, from Assumption \ref{ass:unbounded} (d), $\text{P}(\norm{\mathbf{X}}>t)\leq K/t^s$, then
\begin{eqnarray}
\mathbb{E}[|\ln\norm{\mathbf{X}}|^p]&=&\int_0^\infty \text{P}\left(|\ln \norm{\mathbf{X}}|^p>u\right)du\nonumber\\
&=&\int_0^\infty \left[\text{P}\left(\norm{\mathbf{X}}>e^{u^\frac{1}{p}}\right)+\text{P}\left(\norm{\mathbf{X}}<e^{-u^\frac{1}{p}}\right)\right]du\nonumber\\
&\leq & \int_0^\infty Ke^{-su^\frac{1}{p}}du+\int_0^\infty U_gc_de^{-du^\frac{1}{p}}du\nonumber\\
&\overset{v=su^\frac{1}{p}}{=}&\frac{1}{s^p}\int_0^\infty Kpe^{-v} v^{p-1}dv+\int_0^\infty U_gc_dpe^{-dv}v^{p-1} dv\nonumber\\
&=&\left(\frac{K}{s^p}+\frac{U_gc_d}{d^p}\right)p!.
\label{eq:lnp}
\end{eqnarray}
Using Stirling's formula $p!\leq ep^{p+1/2}e^{-p}$, we have
\begin{eqnarray}
\mathbb{E}[\ln\norm{\mathbf{X}}\mathbf{1}(g(\mathbf{X})\leq t)]&\leq & e^\frac{1}{p} p^{1+\frac{1}{2p}}e^{-1} \left(\frac{K}{s^p}+\frac{U_gc_d}{d^p}\right)^\frac{1}{p}(\mu t^\gamma)^{1-\frac{1}{p}}\nonumber\\
&\lesssim& pt^{\gamma\left(1-\frac{1}{p}\right)},
\label{eq:logpf}
\end{eqnarray}
which holds for all $p>1$. For sufficiently small $t$, let $p=\ln(1/t)$, then the right hand side of \eqref{eq:logpf} becomes $et^\gamma \ln(1/t)$.
\subsection{Proof of Lemma \ref{lem:tail}}\label{sec:tail}
\begin{eqnarray}
\int_{g(\mathbf{x})>t}\frac{f(\mathbf{x})}{g(\mathbf{x})}d\mathbf{x}&=&\mathbb{E}\left[\frac{1}{g(\mathbf{X})}\mathbf{1}(g(\mathbf{X})>t)\right]\nonumber\\
&=&\int_0^\infty \text{P}\left(\frac{1}{g(\mathbf{X})}\mathbf{1}(g(\mathbf{X})>t)>u\right) du\nonumber\\
&=&\int_0^\frac{1}{t} \text{P}\left(g(\mathbf{X})<\frac{1}{u}\right)d\mathbf{u}\nonumber\\
&\leq &\left\{
\begin{array}{ccc}
\frac{\mu}{1-\gamma} t^{\gamma-1} &\text{if} & \gamma<1\\
\mu+\mu\ln \frac{1}{\mu t} &\text{if} & \gamma=1.
\end{array}
\right.
\end{eqnarray}
\subsection{Proof of Lemma \ref{lem:largeradius}}\label{sec:largeradius}
For all $\mathbf{x}\in S_1$,
\begin{eqnarray}
P_g(B(\mathbf{x},a_M))\geq g(\mathbf{x})c_da_M^d-C_1a_M^{d+2}\geq C_1a_M^{d+2}=C_1A^{d+2}M^{-\beta(d+2)}=kM^{-\beta(d+2)}\geq \frac{k}{M},
\end{eqnarray}
in which we used \eqref{eq:S1def} and Lemma \ref{lem:pdf}. Hence, according to \eqref{eq:masslb} and Chernoff inequality, 
\begin{eqnarray}
\text{P}(\nu>a_M|\mathbf{x})&\leq & e^{-MP_g(B(\mathbf{x},a_M))}\left(\frac{eMP_g(B(\mathbf{x},a_M))}{k}\right)^k\nonumber\\
&\leq &e^{-\frac{1}{2}Mg(\mathbf{x})c_da_M^d}\left(\frac{eMg(\mathbf{x})c_da_M^d}{2k}\right)^k\nonumber\\
&:=&\phi(\mathbf{x}).
\label{eq:phidef}
\end{eqnarray}
Moreover, define $a=Mc_da_M^d/2$, then
\begin{eqnarray}
\text{P}(\nu>a_M,\mathbf{X}\in S_1)&=&\left(\frac{e}{k}\right)^k \mathbb{E}\left[e^{-ag(\mathbf{X})}(ag(\mathbf{X}))^k\right]\nonumber\\
&\leq & \left(\frac{e}{k}\right)^k \mathbb{E}\left[e^{-\frac{1}{2}ag(\mathbf{X})}\right]\underset{t>0}{\sup}e^{-\frac{1}{2}t} t^k\nonumber\\
&=&2^k \mathbb{E}\left[e^{-\frac{1}{2}ag(\mathbf{X})}\right]\nonumber\\
&=&2^k \int_0^\infty \text{P}\left(e^{-\frac{1}{2} ag(\mathbf{X})}>u\right)du\nonumber\\
&=&2^k \int_0^\infty \text{P}\left(g(\mathbf{X})<\frac{2}{a}\ln \frac{1}{u}\right)du\nonumber\\
&=&2^{k+\gamma}\mu \int_0^1 \left(\ln \frac{1}{u}\right)^\gamma du\nonumber\\
&=&2^{k+\gamma} \mu \Gamma(\gamma+1)\left(\frac{1}{2}Mc_dA^dM^{-\beta d}\right)^{-\gamma}.
\end{eqnarray}
The proof is complete.
\subsection{Proof of Lemma \ref{lem:lognu}}\label{sec:lognu}
From Assumption \ref{ass:unbounded} (d), $\text{P}(\norm{Y}>r)\leq K/r^s$. Hence $P_g(B^c(\mathbf{0},r))\leq K/r^s$, in which $B^c(\mathbf{0},r)=\mathbb{R}^d\setminus B(\mathbf{0},r)$. Denote $\nu_0$ as the kNN distance of $\mathbf{x}=0$ among $\mathbf{Y}_1,\ldots, \mathbf{Y}_M$. Then for sufficiently large $M$ and $r>(2K)^{1/s}$, we have $P_g(B^c(\mathbf{0},r))\geq 1/2$, hence
\begin{eqnarray}
\text{P}(\nu_0>r)&=&\text{P}\left(n(B^c(\mathbf{0},r))>M-k\right)\nonumber\\
&\leq &\text{P}\left(n(B^c(\mathbf{0},r))>\frac{1}{2}M\right)\nonumber\\
&\leq & e^{-M\frac{K}{r^s}}\left(\frac{eM\frac{K}{r^s}}{\frac{1}{2}M}\right)^{\frac{1}{2}M}\nonumber\\
&\leq & \left(\frac{2eK}{r^s}\right)^{\frac{1}{2}M}.
\end{eqnarray}
Denote $n_Y(S)$ as the number of samples from $\{\mathbf{Y}_1,\ldots, \mathbf{Y}_M \}$ that are in $S$. Then for any given $\mathbf{x}$, and $r \geq (2K)^{1/s}+\norm{\mathbf{x}}$, since $n_Y(B(\mathbf{x},t))\geq n_Y(B(\mathbf{0},t-\norm{\mathbf{x}}))$,
\begin{eqnarray}
\text{P}(\nu>r|\mathbf{x})\leq \left(\frac{2eK}{(r-\norm{\mathbf{x}})^s}\right)^{\frac{1}{2}M}.
\label{eq:nut}
\end{eqnarray}
Let
\begin{eqnarray}
t_0=\max\left\{\ln \frac{2\norm{\mathbf{x}}}{a_M},\frac{1}{s}\ln \frac{2^{1+s}eK}{a_M^s} \right\}.
\end{eqnarray}
It can be checked that $a_Me^{t_0}\geq (2K)^{1/s}+\norm{\mathbf{x}}$, therefore
\begin{eqnarray}
\mathbb{E}\left[\ln \frac{\nu}{a_M}\mathbf{1}(\nu>a_M)|\mathbf{x}\right]&=&\int_0^\infty \text{P}(\nu>a_Me^t|\mathbf{x})dt\nonumber\\
&= &\int_0^{t_0}\text{P}(\nu>a_Me^t|\mathbf{x})dt+\int_{t_0}^\infty \text{P}(\nu>a_Me^t|\mathbf{x})dt\nonumber\\
&\leq &\int_0^{t_0}\text{P}(\nu>a_M|\mathbf{x})dt+\int_{t_0}^\infty \left(\frac{2eK}{(a_Me^t-\norm{\mathbf{x}})^s}\right)^{\frac{1}{2}M}dt\nonumber\\
&\overset{(a)}{\leq}&\phi(\mathbf{x})t_0 +\int_{t_0}^\infty \left(\frac{2^{1+s}eK}{a_M^se^{st}}\right)^{\frac{1}{2}M} dt\nonumber\\
&=&\phi(\mathbf{x})t_0+\left(\frac{2^{1+s}eK}{a_M^s}\right)^{\frac{1}{2}M}\frac{2}{M}e^{-\frac{1}{2}sMt_0}\nonumber\\
&\overset{(b)}{\leq}& \phi(\mathbf{x})t_0+\frac{2}{M}.
\end{eqnarray}
In (a), we use \eqref{eq:phidef} and the definition of $t_0$, which implies that $\norm{\mathbf{x}}\leq a_Me^t/2$. (b) uses the fact that $e^{st_0}\geq 2^{1+s}eK/a_M^s$. Hence
\begin{eqnarray}
\mathbb{E}\left[\ln \frac{\nu}{a_M}\mathbf{1}(\nu>a_M,\mathbf{X}\in S_1)\right]\leq \mathbb{E}[\phi(\mathbf{X})t_0]+\frac{2}{M}.
\end{eqnarray}
It remains to bound $\mathbb{E}[\phi(\mathbf{X})t_0]$. For any $T>0$,
\begin{eqnarray}
\mathbb{E}[\phi(\mathbf{X})t_0]&\leq& \mathbb{E}[\phi(\mathbf{X})t_0\mathbf{1}(t_0\leq T)]+\mathbb{E}[\phi(\mathbf{X})t_0\mathbf{1}(t_0>T)]\nonumber\\
&\leq & T\mathbb{E}[\phi(\mathbf{X})]+\mathbb{E}[t_0\mathbf{1}(t_0>T)].
\end{eqnarray}
In Lemma \ref{lem:largeradius}, we have shown that $\mathbb{E}[\phi(\mathbf{X})]\leq C_3M^{-\gamma(1-\beta d)}$. For the second term,
\begin{eqnarray}
\mathbb{E}[t_0\mathbf{1}(t_0>T)]&\leq& \mathbb{E}\left[\left(\ln \frac{2\norm{\mathbf{X}}}{a_M}+\frac{1}{s}\ln \frac{2^{1+s}eK}{a_M^s}\right)\mathbf{1}\left(\norm{\mathbf{X}}>\frac{1}{2}a_M e^T\right)\right]\nonumber\\
&\leq & \int_0^\infty \text{P}\left(\ln \frac{2\norm{\mathbf{X}}}{a_M}\mathbf{1}\left(\mathbf{X}>\frac{1}{2}a_M e^T\right)>u\right) du+\frac{1}{s}\ln \frac{2^{1+s}eK}{a_M^s}\text{P}\left(\norm{\mathbf{X}}>\frac{1}{2}a_M e^T\right)\nonumber\\
&\leq & \int_0^{T}\text{P}\left(\norm{\mathbf{X}}>\frac{1}{2}a_M e^T\right) du+\int_{T}^\infty \text{P}\left(\norm{\mathbf{X}}>\frac{1}{2}he^u\right) du+\frac{2^sK}{a_M^se^{sT}}\ln \frac{2^{1+s}eK}{a_M^s}\nonumber\\
&\leq & \frac{2^sK}{a_M^se^{sT}s}\left[sT+1+\ln \frac{2^{1+s}eK}{a_M^s}\right].
\end{eqnarray}
Let $T=(1/s)\ln M$, then
\begin{eqnarray}
\mathbb{E}[\phi(\mathbf{X})t_0]\lesssim M^{-\gamma(1-\beta d)}\ln M.
\end{eqnarray}
Hence
\begin{eqnarray}
\mathbb{E}\left[\ln \frac{\nu}{a_M}\mathbf{1}(\nu>a_M,\mathbf{X}\in S_1)\right] \lesssim M^{-\gamma(1-\beta d)}\ln M.
\end{eqnarray}
\subsection{Proof of Lemma \ref{lem:log2}}\label{sec:log2}
\begin{eqnarray}
|\mathbb{E}[\ln \nu \mathbf{1}(\nu<1)|\mathbf{x}]|&=&\int_0^\infty \text{P}(\nu<e^{-t}|\mathbf{x})dt\nonumber\\
&\overset{(a)}{\leq}&\int_0^\infty \text{P}\left(P_g(B(\mathbf{x},\nu))<U_gc_de^{-dt}\right) dt\nonumber\\
&\overset{(b)}{\leq}& \int_0^{\frac{1}{d}\ln \frac{M}{k}}dt+\int_{\frac{1}{d}\ln \frac{M}{k}}^\infty \left(\frac{eMU_gc_de^{-dt}}{k}\right)^k dt\nonumber\\
&=&\frac{1}{d}\ln \frac{M}{k}+\frac{(eU_gc_d)^k}{kd}.
\label{eq:log21}
\end{eqnarray}
In (a), we use Lemma \ref{lem:ub}. (b) uses Chernoff bound. Moreover, let $t_0=\max\{\ln (2\norm{\mathbf{x}}),(1/s)\ln(2^{1+s}eK),0 \}$, then
\begin{eqnarray}
\mathbb{E}[\ln \nu \mathbf{1}(\nu>1)|\mathbf{x}]&=&\int_0^\infty \text{P}(\nu>e^t|\mathbf{x})dt\nonumber\\
&\leq & \int_0^{t_0}dt+\int_{t_0}^\infty \left(\frac{2eK}{(e^t-\norm{\mathbf{X}})^s}\right)^{\frac{1}{2}M} dt\nonumber\\
&=& t_0+\int_{t_0}^\infty \left(\frac{2^{1+s}eK}{e^{st}}\right)^{\frac{1}{2}M} dt\nonumber\\
&=&t_0+(2^{1+s}eK)^{\frac{1}{2}M}\frac{2}{sM}e^{-\frac{1}{2}sMt_0}\nonumber\\
&\leq &\max\left\{\ln (2\norm{\mathbf{x}}),\frac{1}{s}\ln(2^{1+s}eK),0 \right\}+\frac{2}{sM}\nonumber\\
&\leq & |\ln (2\norm{\mathbf{x}})|+\frac{1}{s}|\ln(2^{1+s}eK)|+\frac{2}{sM}.
\label{eq:log22}
\end{eqnarray}
Combining \eqref{eq:log21} and \eqref{eq:log22}, the proof is complete.
\section{Proof of Theorem \ref{thm:var}}\label{sec:var}
From \eqref{eq:estimator}, we have
\begin{eqnarray}
\Var[\hat{D}(f||g)]&=&\Var\left[\frac{d}{N}\sum_{i=1}^N \ln \nu_i-\frac{d}{N}\sum_{i=1}^N \ln \epsilon_i\right]\nonumber\\
&\leq &2\Var\left[\frac{d}{N}\sum_{i=1}^N\ln \epsilon_i\right]+2\Var\left[\frac{d}{N}\sum_{i=1}^N\ln \nu_i\right]\nonumber\\
&:=&2I_1+2I_2.
\end{eqnarray}
We bound $I_1$ and $I_2$ separately.

\textbf{Bound of $I_1$}. $I_1$ is the variance of Kozachenko-Leonenko entropy estimator \cite{kozachenko1987sample}, which estimates $h(f)=-\int f(\mathbf{x})\ln f(\mathbf{x})dx$. Here we use similar proof procedure as was already used in the proof of Theorem 2 in our recent work \cite{zhao2019analysis}. \cite{zhao2019analysis} has analyzed a truncated KL entropy estimator, which means that $\epsilon_i$ is truncated by an upper bound $a_N$. The variance of this estimator is actually equal to $\Var[(d/N)\sum_{i=1}^N \ln \rho_i]$, in which $\rho_i=\min\{\epsilon,a_N \}$. It was shown in \cite{zhao2019analysis} that if $a_N\sim N^{-\beta}$ with $0<\beta<1/d$, then $\Var[(d/N)\sum_{i=1}^N \ln \rho_i]=\mathcal{O}(N^{-1})$. In this section, we prove the same convergence bound for the estimator without truncation, i.e. $\Var[(d/N)\sum_{i=1}^N \ln \epsilon_i]$.

Let $\mathbf{X}_1'$ be a sample that is i.i.d with $\mathbf{X}_1,\mathbf{X}_2,\ldots, \mathbf{X}_N$. Recall that $\epsilon_i$ is the $k$-th nearest neighbor distance of $\mathbf{X}_i$ among $\mathbf{X}_1,\mathbf{X}_2,\ldots,\mathbf{X}_N$. If we replace $\mathbf{X}_1$ with $\mathbf{X}_1'$, then the kNN distances will change. Denote $\epsilon_i'$ as the $k$-th nearest neighbor distance based on $\mathbf{X}_1',\mathbf{X}_2,\ldots, \mathbf{X}_N$. Then use Efron-Stein inequality \cite{tsybakov2009introduction}, 
\begin{eqnarray}
\Var\left[\frac{d}{N}\sum_{i=1}^N\ln \epsilon_i\right]\leq \frac{N}{2}\left[\left(\frac{d}{N}\sum_{i=1}^N \ln \epsilon_i-\frac{d}{N}\sum_{i=1}^N \ln \epsilon_i'\right)^2\right].
\end{eqnarray}
Define $U_i=\ln(Nc_d \epsilon_i^d)$ and $U_i'=\ln(Nc_d (\epsilon_i')^d)$ for $i=1,\ldots,N$. Moreover, define $\epsilon_i''$ as the $k$ nearest neighbor distances based on $\mathbf{X}_2,\ldots, \mathbf{X}_N$, and $U_i''=\ln(Nc_d (\epsilon_i'')^d)$, $i=2,\ldots,N$. Follow the steps in Appendix C of \cite{zhao2019analysis}, we have
\begin{eqnarray}
\Var\left[\frac{d}{N}\sum_{i=1}^N \ln \epsilon_i\right]\leq \frac{2}{N}(2k\gamma_d+1)\left[(k+1)\mathbb{E}[U_1^2]+k\mathbb{E}[(U_1'')^2]\right],
\label{eq:I11}
\end{eqnarray}
in which $\gamma_d$ is a constant that depends on dimension $d$ and the norm we use. For example, if we use $\ell_2$ norm, then $\gamma_d$ is the minimum number of cones with angle $\pi/6$ that cover $\mathbb{R}^d$. 

Now we bound $\mathbb{E}[U_1^2]$ and $\mathbb{E}[(U_1'')^2]$.  Define $\rho=\min\{\epsilon,a_N \}$, in which $a_N\sim N^{-\beta}$, $0<\beta<1/d$. Note that we truncate the estimator for the convenience of analysis, although we are now analyzing an estimator without truncation. The deviation caused by such truncation will be bounded later. In the following proof, we omit the index for convenience. $\mathbb{E}[U^2]$ can be bounded by
\begin{eqnarray}
\mathbb{E}[U^2]&=&\mathbb{E}[(\ln (N\epsilon^d c_d))^2]\nonumber\\
&=&\mathbb{E}\left[\left(\ln (NP_f(B(\mathbf{X},\epsilon)))-\ln \frac{P_f(B(\mathbf{X},\epsilon))}{f(\mathbf{X})c_d\rho^d}+d\ln \frac{\epsilon}{\rho}-\ln f(\mathbf{X})\right)^2\right]\nonumber\\
&\leq & 4\mathbb{E}[(\ln (NP_f(B(\mathbf{X},\epsilon))))^2]+4\mathbb{E}\left[\left(\ln \frac{P_f(B(\mathbf{X},\epsilon))}{f(\mathbf{X})c_d\rho^d}\right)^2\right]+4d^2\mathbb{E}\left[\left(\ln \frac{\epsilon}{\rho}\right)^2\right]+4\mathbb{E}[(\ln f(\mathbf{X}))^2],\nonumber\\
\label{eq:EU2}
\end{eqnarray}
in which $P_f(S)$ is the probability mass of $S$ under a distribution with pdf $f$, i.e. $P_f(S)=\int_S f(\mathbf{x})d\mathbf{x}$.

According to Assumption \ref{ass:var} (b), $\mathbb{E}[(\ln f(\mathbf{X}))^2]=\int f(\mathbf{x})\ln^2 f(\mathbf{x})d\mathbf{x}<\infty$. Moreover, Lemma 6 and Lemma 7 in \cite{zhao2019analysis} have shown that
\begin{eqnarray}
\underset{N\rightarrow \infty }{\lim}\mathbb{E}[(\ln (NP_f(B(\mathbf{X},\epsilon))))^2]=\psi'(k)+\psi^2(k),
\end{eqnarray}
and
\begin{eqnarray}
\underset{N\rightarrow\infty}{\lim}\mathbb{E}\left[\left(\ln \frac{P_f(B(\mathbf{X},\epsilon))}{f(\mathbf{X})c_d\rho^d}\right)^2\right]=0.
\end{eqnarray}
It remains to show that $\mathbb{E}[\ln^2(\epsilon/\rho)]\rightarrow 0$:
\begin{eqnarray}
\mathbb{E}\left[\left(\ln \frac{\epsilon}{\rho}\right)^2\right]&=& \mathbb{E}\left[\left(\ln \frac{\epsilon}{a_N} \right)^2 \mathbf{1}(\epsilon>a_N)\right]\nonumber\\
&\leq & 2\mathbb{E}[\ln^2\epsilon \mathbf{1}(\epsilon>a_N)]+2\mathbb{E}[\ln^2 a_N \mathbf{1}(\epsilon>a_N)]\nonumber\\
&\leq & 2\mathbb{E}[\ln^2\epsilon\mathbf{1}(a_N<\epsilon\leq 1)]+2\mathbb{E}[\ln^2\epsilon \mathbf{1}(\epsilon>1)]+2\ln^2 a_N\text{P}(\epsilon>a_N)\nonumber\\
&\leq & 4\ln^2 a_N\text{P}(\epsilon>a_N)+2\mathbb{E}[\ln^2\epsilon \mathbf{1}(\epsilon>1)].
\label{eq:log3}
\end{eqnarray}
For sufficiently large $N$, $a_N<r_0$. From Assumption \ref{ass:var} (b), for sufficiently small $t$,
\begin{eqnarray}
\text{P}(\tilde{f}(\mathbf{x},a_N)<t)\leq \text{P}\left(\left(\ln \underset{r<r_0}{\inf}\tilde{f}(\mathbf{x},r)\right)^2>\ln^2 t\right)=o\left(\frac{1}{\ln^2 t}\right),
\label{eq:smalltildef}
\end{eqnarray}
in which we use small $o$ notation, since for any variable $U$ such that $U\geq 0$ and $\mathbb{E}[U]<\infty$, $uP(U>u)\rightarrow 0$ as $u\rightarrow \infty$. Since $\beta<1/d$, pick $\delta$ such that $0<\delta<1-\beta d$, then
\begin{eqnarray}
\text{P}(\epsilon>a_N)&\leq& \text{P}\left(P_f(B(\mathbf{X},a_N))<\frac{2k}{N^{1-\delta}}\right)+\text{P}\left(P_f(B(\mathbf{X},\epsilon))\geq \frac{2k}{N^{1-\delta}},\epsilon>a_N\right)\nonumber\\
&\overset{(a)}{\leq} & \text{P}\left(\tilde{f}(\mathbf{x},a_N)<\frac{2k}{N^{1-\delta}c_d a_N^d}\right)+e^{-2kN^\delta}\left(\frac{2ekN^\delta}{k}\right)^k\nonumber\\
&\overset{(b)}{=}& o\left(\frac{1}{(\ln N)^2}\right).
\label{eq:largeeps}
\end{eqnarray}
In (a), we use the definition of $\tilde{f}$ in \eqref{eq:tildedef} for the first term, and use Chernoff inequality for the second term. (b) holds because $N^{1-\delta}a_N^d\sim N^{1-\delta-d\beta}$. $1-\delta-\beta d>0$, thus $N^{1-\delta-d\beta}\rightarrow\infty$. Then we can get \eqref{eq:largeeps} using \eqref{eq:smalltildef}.

Moreover, we can show the following Lemma:
\begin{lem}\label{lem:log1}
	\begin{eqnarray}
	\underset{N\rightarrow 0}{\lim} \mathbb{E}[\ln^2 \epsilon\mathbf{1}(\epsilon>1)]=0.
	\end{eqnarray}
\end{lem}
\begin{proof}
	Please see Appendix~\ref{sec:log1}.
\end{proof}
Based on \eqref{eq:log3}, \eqref{eq:largeeps} and Lemma \ref{lem:log1}, $\mathbb{E}[\ln^2(\epsilon/\rho)]\rightarrow 0$. Therefore \eqref{eq:EU2} becomes
\begin{eqnarray}
\underset{N\rightarrow\infty}{\lim} \mathbb{E}[U^2]\leq 4\left[\psi'(k)+\psi^2(k)+\int f(\mathbf{x})\ln^2 f(\mathbf{x})d\mathbf{x}\right].
\label{eq:EU2lim}
\end{eqnarray}

Similar results hold for $\mathbb{E}[(U'')^2]$. Hence \eqref{eq:I11} becomes
\begin{eqnarray}
\Var\left[\frac{d}{N}\sum_{i=1}^N \ln \epsilon_i\right]=\mathcal{O}\left(\frac{1}{N}\right).
\end{eqnarray}

\textbf{Bound of $I_2$}. Let $\mathbf{Y}_1'$ be a sample that is i.i.d with $\mathbf{Y}_1,\ldots, \mathbf{Y}_M$. Define $\nu_i'$ as the $k$-th nearest neighbor distance of $\mathbf{X}_i$ among $\{\mathbf{Y}_1', \mathbf{Y}_2,\ldots, \mathbf{Y}_M \}$ for $i=1,\ldots,N$. Let $\mathbf{X}_1'$ be a sample that is i.i.d with $\mathbf{X}_1,\ldots, \mathbf{X}_N$, and define $\nu_1''$ as the $k$-th nearest neighbor distance of $\mathbf{X}_1'$ among $\{\mathbf{Y}_1,\ldots, \mathbf{Y}_M \}$. Then from Efron-Stein inequality,
\begin{eqnarray}
I_2&=&\Var\left[\frac{d}{N}\sum_{i=1}^N \ln \nu_i\right]\nonumber\\
&\leq & \frac{M}{2}\mathbb{E}\left[\left(\frac{d}{N}\sum_{i=1}^N \ln \nu_i-\frac{d}{N}\sum_{i=1}^N \ln \nu_i'\right)^2\right]+\frac{N}{2}\mathbb{E}\left[\left(\frac{d}{N}\ln \nu_1-\frac{d}{N}\ln\nu_1''\right)^2\right]\nonumber\\
&=&\frac{Md^2}{2N^2}\mathbb{E}\left[\left(\sum_{i=1}^N(\ln \nu_i-\ln \nu_i')\right)^2\right]+\frac{d^2}{2N}\mathbb{E}[(\ln \nu_1-\ln \nu_1'')^2]\nonumber\\
&:=&I_{21}+I_{22}.
\label{eq:I20}
\end{eqnarray}
To bound the right hand side of \eqref{eq:I20}, we first make the following definitions:
\begin{defi}
	Define two sets $S_1\subset \mathbb{R}^d$, $S_1'\subset \mathbb{R}^d$:
	\begin{eqnarray}
	S_1&:=&\left\{\mathbf{x}|\mathbf{Y}_1 \text{ is among the $k$ neighbors of $\mathbf{x}$ in }\{\mathbf{Y}_1,\ldots, \mathbf{Y}_M\} \right\},\\
	S_1'&:=&\left\{\mathbf{x}|\mathbf{Y}_1' \text{ is among the $k$ neighbors of $\mathbf{x}$ in }\{\mathbf{Y}_1,\ldots, \mathbf{Y}_M\} \right\}.	
	\end{eqnarray}
\end{defi}
\begin{defi}
	Define three events:
	\begin{eqnarray}
	E_1&:&\max\left\{ \underset{i\in [N]}{\max}\norm{\mathbf{X}_i}, \underset{i\in [M]}{\max}\norm{\mathbf{Y}_i}, \norm{\mathbf{Y}_1'} \right\}>(M+N+1)^\frac{5}{s};\\
	E_2&:&\min\left\{\underset{i\in [N]}{\min}\nu_i, \underset{i\in [M]}{\min}\nu_i' \right\}<(M+N)^{-\frac{k+5}{dk}};\label{eq:E2}\\
	E_3&:&\max\{n_X(S_1),n_X(S_1') \}>\frac{2k\gamma_dCN\ln^2 M}{M},
	\end{eqnarray}
	in which $C$ is the constant in Assumption \ref{ass:var} (d). Here, for any set $S$, $n_X(S)=\sum_{i=1}^N \mathbf{1}(\mathbf{X}_i\in S)$ is the number of points from $\mathbf{X}_1,\ldots, \mathbf{X}_N$ that are in $S$, and $\gamma_d$ is the same constant used in \eqref{eq:I11}. 
\end{defi}
We also denote $E=E_1\cup E_2\cup E_3$.

The following lemma shows that all of these three events happen with low probability.
\begin{lem}\label{lem:pe}
	The probabilities of $E_1$, $E_2$ and $E_3$ are bounded by:
	\begin{eqnarray}
	\text{P}(E_1)&\leq& \frac{k}{(M+N+1)^4};\label{eq:pe1}\\
	\text{P}(E_2)&\leq & \left(\frac{eU_gc_d}{k}\right)^k (M+N)^{-4};\label{eq:pe2}\\
	\text{P}(E_3)&\leq & 2\gamma_de^{-k\ln^2 M}(e\ln^2 M)^k+2\exp\left[-(2\ln 2-1)\frac{k\gamma_dCN\ln^2 M}{M}\right].\label{eq:pe3}
	\end{eqnarray}
\end{lem}
\begin{proof}
	Please see Appendix \ref{sec:pe}.
\end{proof}
These three bounds show that $\text{P}(E)\lesssim (M+N)^{-4}$ as long as $N\ln M/M\rightarrow \infty$. Moreover, we show the following lemma:
\begin{lem}\label{lem:log4}
	There exists a constant $C_1$ such that for sufficiently large $M$ we have
	\begin{eqnarray}
	\mathbb{E}[\ln^4 \nu]<C_1\ln^4 M.
	\end{eqnarray}
\end{lem}
\begin{proof}
	Please see Appendix~\ref{sec:log4}.
\end{proof}
Based on Lemma \ref{lem:pe} and Lemma \ref{lem:log4},
\begin{eqnarray}
\mathbb{E}\left[\left(\sum_{i=1}^N(\ln \nu_i-\ln \nu_i')\right)^2\mathbf{1}(E)\right]&\leq & N\mathbb{E}\left[\left(\sum_{i=1}^N(\ln \nu_i-\ln \nu_i')^2\right)\mathbf{1}(E)\right]\nonumber\\
&\leq & 2N\mathbb{E}\left[\sum_{i=1}^N (\ln^2 \nu_i+\ln^2\nu_i')\mathbf{1}(E)\right]\nonumber\\
&= & 4N^2\mathbb{E}[\ln^2\nu \mathbf{1}(E)]\nonumber\\
&\leq &4N^2\sqrt{\mathbb{E}[\ln^4 \nu]\text{P}(E)}\nonumber\\
&\lesssim & \frac{N^2\ln^2 M}{(M+N)^2}.
\label{eq:E}
\end{eqnarray}

If $E$ does not happen, then $\norm{\mathbf{X}_i}$, $\norm{\mathbf{Y}_i}$, $\norm{\mathbf{Y}_1'}$ are all upper bounded by $(M+N+1)^{(5/s)}$. Thus $\nu_i$ and $\nu_i'$ are all upper bounded by $2(M+N+1)^{(5/s)}$. Besides, from \eqref{eq:E2}, they are both lower bounded by $(M+N)^{-\frac{k+5}{dk}}$. There are at most $n_X(S_1)+n_X(S_1')$ points such that $\nu_i\neq \nu_i'$. Hence
\begin{eqnarray}
\mathbb{E}\left[\left(\sum_{i=1}^N(\ln \nu_i-\ln \nu_i')\right)^2\mathbf{1}(E^c)\right]&\leq& \left[\frac{4k\gamma_dCN\ln^2 M}{M}\left(\frac{5}{s}\ln((2(M+N+1)))+\frac{k+5}{dk}\ln(M+N)  \right)\right]^2\nonumber\\
&\lesssim &\frac{N^2}{M^2}\ln^4 M\ln^2(M+N).
\label{eq:Ec}
\end{eqnarray}
Combining \eqref{eq:E} and \eqref{eq:Ec}, we have
\begin{eqnarray}
I_{21}\lesssim \frac{\ln^4 M\ln^2(M+N)}{M}.
\end{eqnarray}

Then $I_{22}$ can be bounded by:
\begin{eqnarray}
I_{22}&=&\frac{1}{2N}\mathbb{E}[(\ln(Nc_d\nu_1^d)-\ln(Nc_d(\nu_1')^d))^2]\nonumber\\
&\leq & \frac{1}{N}\left[\mathbb{E}[(\ln(Mc_d\nu_1^d))^2]+\mathbb{E}[(\ln(Mc_d(\nu_1')^d))^2]\right]\nonumber\\
&=&\frac{2}{N}\mathbb{E}[(\ln(Mc_d\nu_1^d))^2].
\end{eqnarray}
Similar to the analysis from \eqref{eq:EU2} to \eqref{eq:EU2lim}, we can show that the limit of $\mathbb{E}[(\ln(Mc_d\nu_1^d))^2]$ can also be bounded by the right hand side of \eqref{eq:EU2lim}. Therefore
\begin{eqnarray}
I_{22}\lesssim \frac{1}{N},
\end{eqnarray}
\begin{eqnarray}
I_2=I_{21}+I_{22}\lesssim \frac{1}{N}+\frac{\ln^4 M\ln^2(M+N)}{M},
\end{eqnarray}
and
\begin{eqnarray}
\Var[\hat{D}(f||g)]\leq 2I_1+2I_2\lesssim \frac{1}{N}+\frac{\ln^4 M\ln^2(M+N)}{M}.
\end{eqnarray}

\subsection{Proof of Lemma \ref{lem:log1}}\label{sec:log1}
Similar to \eqref{eq:nut}, we can show that for any given $\mathbf{x}$, and $t\geq(2K)^{1/s}+\norm{\mathbf{x}}$,
\begin{eqnarray}
\text{P}(\epsilon>t|\mathbf{x})\leq \left(\frac{2eK}{(t-\norm{\mathbf{x}})^s}\right)^{\frac{1}{2}(N-1)}.
\label{eq:epst}
\end{eqnarray}
Then
\begin{eqnarray}
\mathbb{E}[\ln^2 \epsilon\mathbf{1}(\epsilon>1)]&=&\int_0^\infty \text{P}\left(\ln^2\epsilon \mathbf{1}(\epsilon>1)>t\right)dt\nonumber\\
&=&\int_0^\infty \text{P}(\epsilon>e^{\sqrt{t}})dt.
\end{eqnarray}
Therefore if $(1/2)e^{\sqrt{t}}\geq (2K)^{1/s}$,
\begin{eqnarray}
\text{P}(\epsilon>e^{\sqrt{t}})&\leq& \text{P}\left(\norm{\mathbf{X}}>\frac{1}{2}e^{\sqrt{t}}\right)+\text{P}\left(\norm{\mathbf{X}}<\frac{1}{2}e^{\sqrt{t}}, \epsilon>e^{\sqrt{t}}\right)\nonumber\\
&\leq &\frac{k}{\left(\frac{1}{2}e^{\sqrt{t}}\right)^s}+\left(\frac{2eK}{\left(e^{\sqrt{t}}-\frac{1}{2}e^{\sqrt{t}}\right)^s}\right)^{\frac{1}{2}(N-1)}\nonumber\\
&=&2^s Ke^{-\frac{1}{2} st}+(2^{1+s}eK)^{\frac{1}{2}(N-1)}e^{-\frac{1}{2} s(N-1) t}.
\end{eqnarray}
Define
\begin{eqnarray}
\phi(t)=\left\{
\begin{array}{ccc}
1 &\text{if}& t\leq \max\left\{\ln^2(2^{1+\frac{1}{s}}K^\frac{1}{s}),\frac{2}{s}\ln(2^{1+s}eK) \right\} \\
2^sKe^{-\frac{1}{2}st}+e^{-\frac{1}{4}st} &\text{if} & t>\max\left\{\ln^2(2^{1+\frac{1}{s}}K^\frac{1}{s}),\frac{2}{s}\ln(2^{1+s}eK) \right\}.
\end{array}
\right.
\end{eqnarray}
It can be shown that $\text{P}(\epsilon>e^{\sqrt{t}})\leq \phi(t)$. Since $\phi(t)$ is integrable in $(0,\infty)$, according to Lebesgue dominated convergence theorem, 
\begin{eqnarray}
\underset{N\rightarrow\infty}{\lim}\mathbb{E}[\ln^2\epsilon\mathbf{1}(\epsilon>1)]=\int_0^\infty \underset{N\rightarrow\infty}{\lim}\text{P}(\epsilon>e^{\sqrt{t}})dt=0.
\end{eqnarray}

\subsection{Proof of Lemma \ref{lem:pe}}\label{sec:pe}
\textbf{Proof of \eqref{eq:pe1}}. According to Assumption \ref{ass:var} (c), for $i=1,\ldots,N$,
\begin{eqnarray}
\text{P}(\norm{\mathbf{X}_i}>t)\leq \frac{\mathbb{E}[\norm{\mathbf{X}_i}^s]}{t^s}\leq \frac{K}{t^s}.
\end{eqnarray}
Similar bound holds for $\norm{\mathbf{X}_1'}$ and $\mathbf{Y}_i$, $i=1,\ldots,M$. Let $t=(M+N+1)^{(5/s)}$, and using the union bound, we get \eqref{eq:pe1}.

\textbf{Proof of \eqref{eq:pe2}}. Since $g$ is bounded by $U_g$, we have $P_g(B(\mathbf{x},r))\leq U_gc_dr^d$ for any $\mathbf{x}$ and $r>0$. Let $r_0=(M+N)^{-\frac{k+5}{dk}}$, then for sufficiently large $M$,
\begin{eqnarray}
U_gc_dr_0^d<\frac{k}{M}.
\end{eqnarray}
Hence from Chernoff inequality,
\begin{eqnarray}
\text{P}(\nu_i<r_0)\leq \exp[-MU_gc_dr_0^d]\left(\frac{eMU_gc_dr_0^d}{k}\right)^k\leq \left(\frac{eMU_gc_dr_0^d}{k}\right)^k.
\end{eqnarray}
Then \eqref{eq:pe2} can be obtained by calculating the union bound.

\textbf{Proof of \eqref{eq:pe3}}. We first prove \eqref{eq:pe3} under the condition that we are using $\ell_2$ norm first. We will then generalize the result to the case with arbitrary norm. Define
\begin{eqnarray}
S(\mathbf{y})&:=&\left\{\mathbf{x}|\mathbf{y} \text{ is among the $k$ neighbors of $\mathbf{x}$ in }\{\mathbf{Y}_1,\ldots, \mathbf{Y}_M\} \right\},
\end{eqnarray}
then $S_1=S(\mathbf{Y}_1)$, $S_1'=S(\mathbf{Y}_1')$. 

Recall that $\gamma_d$ is defined as the minimum number of cones with angle $\pi/6$ that can cover $\mathbb{R}^d$. Now we pick any $\mathbf{y}\in \mathbb{R}^d$, and divide $\mathbb{R}^d$ into $\gamma_d$ cones with angle $\pi/6$, such that $\mathbf{y}$ is the vertex of all the cones. These cones are named as $C_j$, $j=1,\ldots, \gamma_d$, and then $\cup_{j=1}^{\gamma_d} C_j=\mathbb{R}^d$. Define $r_j$ such that
\begin{eqnarray}
P_g(B(\mathbf{y},r_j)\cap C_j)=\frac{k\ln^2 M}{M}, \forall j=1,\ldots, \gamma_d.
\end{eqnarray}
Define $n_Y(S)=\sum_{i=1}^M \mathbf{1}(\mathbf{Y_i}\in S)$ as the number of points from $\{\mathbf{Y}_1,\ldots, \mathbf{Y}_M \}$ that are in $S$. Moreover, define
\begin{eqnarray}
S_0(\mathbf{y})=\cup_{j=1}^{\gamma_d} B(\mathbf{y},r_j)\cap C_j.
\end{eqnarray}
Then from Chernoff inequality,
\begin{eqnarray}
\text{P}(n_Y(B(\mathbf{y},r_j)\cap C_j)<k)\leq e^{-k\ln^2 M}(e\ln^2 M)^k,
\end{eqnarray} 
and
\begin{eqnarray}
\text{P}\left(\cup_{j=1}^{\gamma_d}n_Y(B(\mathbf{y},r_j)\cap C_j)<k\right)\leq \gamma_de^{-k\ln^2 M}(e\ln^2 M)^k.
\end{eqnarray}
This result indicates that with probability at least $1-\gamma_de^{-k\ln^2 M}(e\ln^2 M)^k$, there are at least $k$ points in $B(\mathbf{y}, r_j)\cap C_j$ for $j=1,\ldots, \gamma_d$. 

Under this condition, we can show that $S(\mathbf{y})\subset S_0(\mathbf{y})$. For any $\mathbf{x}\notin S_0(\mathbf{y})$, since $\cup_{j=1}^{\gamma_d}C_j=\mathbb{R}^d$, $\mathbf{x}\in C_j$ for some $j\in \{1,\ldots,\gamma_d \}$. In $B(\mathbf{y},r_j)\cap C_j$, there are already at least $k$ points, $\mathbf{Y}_{i_l}$, $l=1,\ldots, k$, among $\mathbf{Y}_1,\ldots, \mathbf{Y}_M$. Then $\norm{\mathbf{Y}_{i_l}-\mathbf{y}}<r_j$ for $\l=1,\ldots, k$, while $\norm{\mathbf{x}-\mathbf{y}}\geq r_j$. Denote $\theta$ as the angle between vector $\mathbf{Y}_{i_l}-\mathbf{y}$ and $\mathbf{x}-\mathbf{y}$. Since $\mathbf{Y}_{i_l}\in C_j$ and $\mathbf{x}\in C_j$, we have $\theta<\pi/3$, and thus
\begin{eqnarray}
\norm{\mathbf{Y}_{i_l}-\mathbf{x}}^2&=&\norm{\mathbf{x}-\mathbf{y}}^2+\norm{\mathbf{Y}_{i_l}-\mathbf{y}}^2-2\norm{\mathbf{x}-\mathbf{y}}\norm{\mathbf{Y}_{i_l}-\mathbf{y}}\cos \theta\nonumber\\
&<&\norm{\mathbf{x}-\mathbf{y}}^2+\norm{\mathbf{Y}_{i_l}-\mathbf{y}}^2-\norm{\mathbf{x}-\mathbf{y}}\norm{\mathbf{Y}_{i_l}-\mathbf{y}}\nonumber\\
&<&\norm{\mathbf{x}-\mathbf{y}},
\end{eqnarray}
which indicates that $\norm{\mathbf{y}-\mathbf{x}}>\norm{\mathbf{Y}_{i_l}-\mathbf{x}}$ for $l=1,\ldots, k$. $\mathbf{Y}_{i_l}$, $l=1,\ldots, k$ are all closer to $\mathbf{x}$ than $\mathbf{y}$, therefore $\mathbf{y}$ can not be one of the $k$ nearest neighbors of $\mathbf{x}$, i.e. $\mathbf{x}\notin S(\mathbf{y})$. Recall that $\mathbf{x}$ is arbitrarily picked outside $S_0(\mathbf{y})$, thus $S(\mathbf{y})\subset S_0(\mathbf{y})$.  Therefore with probability at least $1-\gamma_de^{-k\ln^2 M}(e\ln^2 M)^k$,
\begin{eqnarray}
P_g(S(\mathbf{y}))\leq P_g(S_0(\mathbf{y}))=P_g\left(\cup_{j=1}^{\gamma_d} B(\mathbf{y},r_j)\cap C_j\right)\leq \frac{k\gamma_d \ln^2 M}{M}.
\end{eqnarray}

Using Assumption \ref{ass:var} (d), $P_f(S)\leq CP_g(S)$ for any $S\subset \mathbb{R}^d$. If both $S(\mathbf{Y}_1)\subset S_0(\mathbf{Y}_1)$ and $S(\mathbf{Y}_1)\subset S_0(\mathbf{Y}_1)$ hold, then
\begin{eqnarray}
\max\{P_f(S_1),P_f(S_1') \}\leq \frac{k\gamma_dC\ln^2M}{M}.
\end{eqnarray}
Using Chernoff inequality again,
\begin{eqnarray}
\text{P}\left(n_X(S_1)>\frac{2k\gamma_d CN\ln^2M}{M}|S(\mathbf{Y}_1)\subset S_0(\mathbf{Y}_1)\right)&\leq& e^{-NP_f(S_1)}\left(\frac{eNP_f(S_1)}{2k\gamma_dCN\ln^2 M\frac{1}{M}}\right)^\frac{2k\gamma_dCN\ln^2M}{M}\nonumber\\
&\leq &\exp\left[-(2\ln 2-1)k\gamma_d C\ln^2 M\frac{N}{M}\right].
\end{eqnarray}
Therefore
\begin{eqnarray}
\text{P}(E_3)&\leq& \text{P}(S(\mathbf{Y}_1)\not\subset S_0(\mathbf{Y}_1))+\text{P}\left(n_X(S_1)>\frac{2k\gamma_d CN\ln^2M}{M}|S(\mathbf{Y}_1)\subset S_0(\mathbf{Y}_1)\right)\nonumber\\
&&\hspace{-3mm}+\text{P}(S(\mathbf{Y}_1')\not\subset S_0(\mathbf{Y}_1'))+\text{P}\left(n_X(S_1')>\frac{2k\gamma_d CN\ln^2M}{M}|S(\mathbf{Y}_1')\subset S_0(\mathbf{Y}_1')\right)\nonumber\\
&\leq & 2\gamma_de^{-k\ln^2 M}(e\ln^2 M)^k+2\exp\left[-(2\ln 2-1)\frac{k\gamma_dCN\ln^2 M}{M}\right].
\end{eqnarray}
The proof is complete.
\subsection{Proof of Lemma \ref{lem:log4}}\label{sec:log4}
Define
\begin{eqnarray}
t_1=\max\left\{\ln^4 (2\norm{\mathbf{x}}), \frac{16}{s^4}\ln^4(2^{1+s}eK) \right\},
\label{eq:t1def}
\end{eqnarray}
and
\begin{eqnarray}
t_2=\left(\frac{2}{d}\ln \frac{MU_gc_d}{k}\right)^4,
\label{eq:t2def}
\end{eqnarray}
then
\begin{eqnarray}
\mathbb{E}[\ln^4\nu|\mathbf{x}]&=&\int_0^\infty \text{P}\left(\ln^4 \nu>t|\mathbf{x}\right)dt\nonumber\\
&=& \int_0^\infty \text{P}\left(\nu>e^{t^\frac{1}{4}}|\mathbf{x}\right) dt+\int_0^\infty \text{P}\left(\nu<e^{-t^\frac{1}{4}}|\mathbf{x}\right) dt.
\label{eq:log4-1}
\end{eqnarray}
\begin{eqnarray}
\int_0^\infty \text{P}\left(\nu>e^{t^\frac{1}{4}}|\mathbf{x}\right) dt&\leq & \int_0^{t_1} dt+\int_{t_1}^\infty \left(\frac{2eK}{(e^{t^\frac{1}{4}}-\norm{\mathbf{x}})^s}\right)^{\frac{1}{2}M} dt\nonumber\\
&\overset{(a)}{\leq} & t_1+\int_{t_1}^\infty \left(\frac{2^{1+s}eK}{e^{st^\frac{1}{4}}}\right)^{\frac{1}{2}M} dt\nonumber\\
&\overset{u=t^\frac{1}{4}}{=}& t_1+(2^{1+s}eK)^{\frac{1}{2}M}\int_{t_1^\frac{1}{4}}^\infty e^{-\frac{1}{2}sMu} 4u^3 du \nonumber\\
&=&t_1+(2^{1+s}eK)^{\frac{1}{2}M}\frac{256}{s^3M^3}\left(\underset{u}{\sup}\left(\frac{1}{4}sMu\right)^3 e^{-\frac{1}{4}sMu}\right)\int_{t_1^\frac{1}{4}}^\infty e^{-\frac{1}{4}Mu} du\nonumber\\
&=&t_1+(2^{1+s}eK)^{\frac{1}{2}M} \frac{27648e^{-3}}{s^4M^4}\exp\left[-\frac{1}{4}sMt_1^\frac{1}{4}\right]\nonumber\\
&\overset{(b)}{\leq} & \ln^4 (2\norm{\mathbf{x}})+ \frac{16}{s^4}\ln^4(2^{1+s}eK)+\frac{27648e^{-3}}{s^4M^4}\nonumber\\
&\lesssim &\ln^4 \norm{\mathbf{x}}+1.
\end{eqnarray}
In (a), we use $\norm{\mathbf{x}}<e^{t^\frac{1}{4}}/2$. This is true because of the definition of $t_1$ in \eqref{eq:t1def}. (b) holds because according to \eqref{eq:t1def}, $(2^{1+s}eK)^{\frac{1}{2}M} \exp\left[-\frac{1}{4}sMt_1^\frac{1}{4}\right]<1$.

Now we bound the second term in \eqref{eq:log4-1}. Using Chernoff inequality,
\begin{eqnarray}
\int_0^\infty \text{P}\left(\nu<e^{-t^\frac{1}{4}} \right)dt&\leq& \int_0^\infty \text{P}\left(P(B(\mathbf{x},\nu))<U_gc_de^{-dt^\frac{1}{4}}\right) dt\nonumber\\
&=&\int_0^{t_2} dt+\int_{t_2}^\infty \left(\frac{eMU_gc_d\exp[-dt^\frac{1}{4}]}{k}\right)^k dt \nonumber\\
&\lesssim & \ln^4 M.
\end{eqnarray}
Thus
\begin{eqnarray}
\mathbb{E}[\ln^4 \nu]\lesssim 1+\ln^4 M+\mathbb{E}[\ln^4\norm{\mathbf{X}}]\sim \ln^4 M,
\end{eqnarray}
in which the last step uses \eqref{eq:lnp}.

\section{Proof of Theorem \ref{thm:mmx1}}\label{sec:mmx1}
In this section, we show the minimax convergence rate of KL divergence estimator for distributions with bounded support and densities bounded away from zero. The proof can be divided into proving the following three bounds separately:

\begin{eqnarray}
R_a(N,M)&\gtrsim& \frac{1}{M}+\frac{1}{N},\label{eq:Ra1}\\
R_a(N,M)&\gtrsim & N^{-\frac{2}{d}\left(1+\frac{2}{\ln\ln N}\right)}\ln^{-2}N \ln^{-\left(2-\frac{2}{d}\right)}(\ln N),\label{eq:Ra2}\\
R_a(N,M)&\gtrsim & M^{-\frac{2}{d}\left(1+\frac{2}{\ln\ln M}\right)}\ln^{-2}M \ln^{-\left(2-\frac{2}{d}\right)}(\ln M).\label{eq:Ra3}
\end{eqnarray}

\textbf{Proof of \eqref{eq:Ra1}}.

Let $\mathbf{X}$ be supported on $[0,1]^d$, and
\begin{eqnarray}
f_1(\mathbf{x})=\left\{
\begin{array}{ccc}
\frac{3}{2} &\text{if} & 0\leq x_1\leq \frac{1}{2}\\
\frac{1}{2} &\text{if} & \frac{1}{2}<x_1\leq 1,
\end{array}
\right.
f_2(\mathbf{x})=\left\{
\begin{array}{ccc}
\frac{3}{2}+\delta &\text{if} & 0\leq x_1\leq \frac{1}{2}\\
\frac{1}{2}-\delta &\text{if} & \frac{1}{2}<x_1\leq 1,
\end{array}
\right.
\end{eqnarray}
and $g(\mathbf{x})=1$. Then 
\begin{eqnarray}
D(f_1||g)=\int f_1(\mathbf{x})\ln \frac{f_1(\mathbf{x})}{g(\mathbf{x})}d\mathbf{x}=\frac{3}{4}\ln \frac{3}{2}+\frac{1}{4}\ln \frac{1}{2},
\end{eqnarray}
\begin{eqnarray}
D(f_2||g)&=&\left(\frac{3}{4}+\frac{1}{2}\delta\right)\ln \left(\frac{3}{2}+\delta\right)+\left(\frac{1}{4}-\frac{1}{2}\delta\right)\ln \left(\frac{1}{2}-\delta\right)\nonumber\\
&=&\frac{3}{4}\ln \frac{3}{2}+\frac{1}{4}\ln \frac{1}{2}+\left(\frac{1}{2}\ln 3\right)\delta +\mathcal{O}(\delta^2).
\end{eqnarray}
Therefore, for sufficiently small $\delta$, $D(f_2||g)-D(f_1||g)\geq (\ln 3)\delta/4$. Moreover,
\begin{eqnarray}
D(f_1||f_2)=-\frac{3}{4}\ln \left(1+\frac{2}{3}\delta\right)-\frac{1}{4}\ln (1-2\delta).
\end{eqnarray}

By Taylor expansion, it can be shown that $\ln (1+2\delta/3)\geq 2\delta/3-\delta^2/9$, and $\ln (1-2\delta)\geq -2\delta+2\delta^2$, thus
\begin{eqnarray}
D(f_1||f_2)\leq \frac{2}{3}\delta^2.
\end{eqnarray}

Therefore, from Le Cam's lemma \cite{tsybakov2009introduction},
\begin{eqnarray}
R_a(N,M)&\geq& \frac{1}{4}\left(D(f_1||g)-D(f_2||g)\right)^2\exp[-ND(f_1||f_2)]\nonumber\\
&\geq &\frac{1}{4}\left(\frac{1}{4}\ln 3\right)^2\delta^2\exp\left[-\frac{2}{3}N\delta^2\right].
\end{eqnarray}

Let $\delta=1/\sqrt{N}$, then 
\begin{eqnarray}
R_a(N,M)\gtrsim \frac{1}{N}.
\end{eqnarray}

Similarly, let
\begin{eqnarray}
g_1(\mathbf{x})=\left\{
\begin{array}{ccc}
\frac{3}{2} &\text{if} & 0\leq x_1\leq \frac{1}{2}\\
\frac{1}{2} &\text{if} & \frac{1}{2}<x_1\leq 1,
\end{array}
\right.
g_2(\mathbf{x})=\left\{
\begin{array}{ccc}
\frac{3}{2}+\delta &\text{if} & 0\leq x_1\leq \frac{1}{2}\\
\frac{1}{2}-\delta &\text{if} & \frac{1}{2}<x_1\leq 1,
\end{array}
\right.
f(\mathbf{x})=1,
\end{eqnarray}
for $\mathbf{x}\in [0,1]^d$. Then it can be shown that
\begin{eqnarray}
R_a(N,M)\gtrsim \frac{1}{M}.
\end{eqnarray}
The proof of \eqref{eq:Ra1} is complete.

\textbf{Proof of \eqref{eq:Ra2}}.

The proof has similar idea with \cite{wu2016minimax} and \cite{zhao2019analysis}. To begin with, define
\begin{eqnarray}
\mathcal{F}_a&&=\left\{(f,g)|f(\mathbf{x})=(1-\alpha)Q_a(\mathbf{x})+\sum_{i=1}^m \frac{u_i}{mD^d}Q_a\left(\frac{\mathbf{x}-\mathbf{a}_i}{D}\right)\right.,\nonumber\\  &&\hspace{18mm}g(\mathbf{x})=(1-\alpha)Q_a(\mathbf{x})
+\sum_{i=1}^m \frac{\alpha}{mD^d}Q_a\left(\frac{\mathbf{x}-\mathbf{a}_i}{D}\right),\nonumber\\
&&\hspace{18mm}\left.\frac{1}{m}\sum_{i=1}^m u_i=\alpha, 1<mD^{d-1}<C_1, \frac{u_i}{mD^d}\in \{0 \}\cup (c,1)\right\},\nonumber\\
\label{eq:fa}
\end{eqnarray}
in which $Q_a(\mathbf{x})=1/v_d$ for $\mathbf{x}\in B(\mathbf{0},1)$, $v_d$ is the unit ball volume, thus $\int Q_a(\mathbf{x})d\mathbf{x}=1$. $C_1$ and $c$ are two constants. $\alpha\in (0,1)$ and $D$ decrease with $N$, while $m$ increases with $N$. $\mathbf{a}_i$, $i=1,\ldots,n$ are selected such that $\norm{\mathbf{a}_i-\mathbf{a}_j}>2D$ for all $i,j\in \{1,\ldots,m\}$ and $i\neq j$. It can be checked that both $f$ and $g$ integrate to $1$. The condition $u_i/(mD^d)\in \{0\}\cup(c,1)$ is designed such that the density in the support is bounded away from zero, i.e. if $f(\mathbf{x})>0$, then $f(\mathbf{x})\geq c$. Moreover, the surface area of the support is $s_d(1+mD^{d-1})$, in which $s_d$ is the surface area of unit ball, and $s_d=dv_d$. With the condition $1<mD^{d-1}<C_1$, the surface area of the supports of $f$ and $g$ are both upper bounded by $s_dC_1$. Therefore, for sufficiently large $H_f$, $H_g$, $U_f$, $U_g$ and sufficiently small $L_f$ and $L_g$, $\mathcal{F}_a\in \mathcal{S}_a$. Define
\begin{eqnarray}
R_{a1}(N,M)=\underset{\hat{D}}{\inf}\underset{(f,g)\in \mathcal{F}_a}{\sup}\mathbb{E}\left[(\hat{D}(N,M)-D(f||g))^2 \right].
\end{eqnarray}
Recall that $R_a(N,M)$ is defined as the minimax mean square error over $\mathcal{S}_a$, hence
\begin{eqnarray}
R_a(N,M)\geq R_{a1}(N,M).
\label{eq:0to1}
\end{eqnarray}

To derive a lower bound of $R_{a1}(N,M)$, we use Le Cam's method again, with Poisson sampling. Define
\begin{eqnarray}
R_{a2}=\underset{\hat{D}}{\inf}\underset{(f,g)\in \mathcal{F}_a}{\sup}\mathbb{E}\left[(\hat{D}(N',M)-D(f||g))^2 \right],
\end{eqnarray}
in which $N'\sim \text{Poi}(N)$, Poi is the Poisson distribution. Then we have the following lemma:
\begin{lem}\label{lem:1to2}
	\begin{eqnarray}
	R_{a1}(N,M)\geq R_{a2}(2N,M)-\frac{1}{4}\exp[-(1-\ln 2)N].
	\label{eq:1to2}
	\end{eqnarray}
\end{lem}
\begin{proof}
	Please refer to Appendix~\ref{sec:1to2} for details.
\end{proof}
Furthermore, define
\begin{eqnarray}
\mathcal{F}_a'&=&\left\{(f,g)|f(\mathbf{x})=(1-\alpha)Q_a(\mathbf{x})+\sum_{i=1}^m \frac{u_i}{mD^d}Q_a\left(\frac{\mathbf{x}-\mathbf{a}_i}{D}\right),\right.\nonumber\\
&&\hspace{15mm} g(\mathbf{x})=(1-\alpha)Q_a(\mathbf{x})
+\sum_{i=1}^m \frac{\alpha}{mD^d}Q_a\left(\frac{\mathbf{x}-\mathbf{a}_i}{D}\right),\\&&\hspace{15mm}\left.\left|\frac{1}{m}\sum_{i=1}^m u_i-\alpha\right|<\epsilon, 1<mD^{d-1}<C_1, \frac{u_i}{mD^d}\in \{0 \}\cup (c(1+\epsilon),1-\epsilon) \right\}.\nonumber\\
\label{eq:fa2}
\end{eqnarray}
Comparing with the definition of $\mathcal{F}_a$ in \eqref{eq:fa}, the only difference is that we now allow $(1/m)\sum_{i=1}^m u_i$ to deviate slightly from $\alpha$. As a result, $f$ is not necessarily a pdf, since it is not normalized. However, we extend the definition of KL divergence $D(f||g)=\int f(\mathbf{x})\ln (f(\mathbf{x})/g(\mathbf{x}))d\mathbf{x}$ here. Define
\begin{eqnarray}
R_{a3}(N,M,\epsilon)=\underset{D}{\inf}\underset{(f,g)\in \mathcal{F}_a'}{\sup}\mathbb{E}[(\hat{D}(N',M)-D(f||g))^2],
\end{eqnarray}
in which $N'\sim \text{Poi}(N\int f(\mathbf{x})d\mathbf{x})$. Then the number of samples falling on any two disjoint intervals are mutually independent. $R_{a2}$ can be lower bounded by $R_{a3}$ with the following lemma:
\begin{lem}\label{lem:2to3}
	If $\epsilon<\alpha/2$, then
	\begin{eqnarray}
	R_{a2}((1-\epsilon)N,M)\geq \frac{1}{2}R_{a3}(N,M)-3\epsilon^2\left(\ln^2\frac{\alpha}{mD^dv_d}+\ln^2\alpha+\frac{9}{4}\right).
	\end{eqnarray}
\end{lem}
\begin{proof}
	Please refer to Appendix~\ref{sec:2to3} for details.
\end{proof}
With Lemma \ref{lem:1to2} and Lemma \ref{lem:2to3}, the problem of bounding $R_a(N,M)$ can be converted to bounding $R_{a3}(M,N,\epsilon)$. We then show the following lemma, which is slightly modified from Lemma 11 in \cite{zhao2019analysis}.
\begin{lem}\label{lem:lecam}
	Let $U$, $U'$ be two random variables that satisfy the following conditions:
	
	1) $U,U'\in [\eta \lambda, \lambda]$, in which $\lambda\leq (1-\epsilon)mD^d$, $0<\eta<1$, and $\eta \lambda\geq c(1+\epsilon)mD^d$;
	
	2) $\mathbb{E}[U]=\mathbb{E}[U']=\alpha$.
	
	Define
	\begin{eqnarray}
	\Delta=\left|\mathbb{E}\left[U\ln \frac{1}{U}\right]-\mathbb{E}\left[U'\ln \frac{1}{U'}\right]\right|.
	\end{eqnarray}
	Let
	\begin{eqnarray}
	\epsilon=4\lambda/\sqrt{m},
	\label{eq:eps}
	\end{eqnarray}
	then
	\begin{eqnarray}
	R_{a3}(N,M,\epsilon)&\geq& \frac{\Delta^2}{16}\left[\frac{31}{32}-\frac{64\lambda^2\left(\ln \frac{m}{\lambda}\right)^2}{m\Delta^2}-m\mathbb{TV}\left(\mathbb{E}\left[\text{Poi}\left(\frac{NU}{m}\right)\right],\mathbb{E}\left[\text{Poi}\left(\frac{NU'}{m}\right)\right]\right)\right.\nonumber\\
	&&-\left.\frac{16\lambda^2}{m\Delta^2}(d\ln D+h(Q_a))^2\right],
	\label{eq:lecam}
	\end{eqnarray}
	in which $h(Q_a)=\ln v_d$ is the differential entropy of $Q_a$.	
\end{lem}
\begin{proof}
	The proof is exactly the same as the proof of Lemma 11 in \cite{zhao2019analysis}. Condition (1) is different from the corresponding condition in \cite{zhao2019analysis}, but such difference does not affect the proof.
\end{proof}

We construct $U$, $U'$ as following. Let $X, X'\in [\eta,1]$ have matching moments to the $L$-th order, and let
\begin{eqnarray}
P_U(du)&=&\left(1-\mathbb{E}\left[\frac{\eta}{X}\right]\right)\delta_0(du)+\frac{\alpha}{u}P_{\alpha X/\eta}(du),\\
P_{U'}(du)&=&\left(1-\mathbb{E}\left[\frac{\eta}{X'}\right]\right)\delta_0(du)+\frac{\alpha}{u}P_{\alpha X'/\eta}(du),
\end{eqnarray} 
in which $\delta_0$ denotes the distribution that puts all the mass on $u=0$. Now we assume $\alpha\leq (1-\epsilon)mD^d\eta$. Let $\lambda=\alpha/\eta$, then $U, U'$ are supported in $[0,\lambda]$, and condition (1) in Lemma \ref{lem:lecam} is satisfied. Then from Lemma 4 in \cite{wu2016minimax},
\begin{eqnarray}
\Delta&=&\mathbb{E}\left[U\ln \frac{1}{U}-U'\ln \frac{1}{U'}\right]\nonumber\\
&=&\alpha \left(\mathbb{E}\left[\ln \frac{1}{X}\right]-\mathbb{E}\left[\ln \frac{1}{X'}\right]\right),
\end{eqnarray}
and $\mathbb{E}[U^j]=\mathbb{E}[U'^j]$ for $j=1,\ldots,L$. In particular, $\mathbb{E}[U]=\mathbb{E}[U']=\alpha$. When $X$ and $X'$ are properly selected, according to eq.(34) in \cite{wu2016minimax},
\begin{eqnarray}
\left|\mathbb{E}\left[\ln \frac{1}{X}\right]-\mathbb{E}\left[\ln \frac{1}{X'}\right]\right|=2\underset{p\in \mathcal{P}_L}{\inf}\underset{x\in [\eta,1]}{\sup}|\ln x-p(x)|,
\end{eqnarray}
in which $\mathcal{P}_L$ is the set of all polynomials with degree $L$.

According to eq.(5) and (6) in page 445 in \cite{timan1963theory}, for $a>1$, $L\rightarrow \infty$,
\begin{eqnarray}
\underset{p\in \mathcal{P}_L}{\inf}\underset{t\in [-1,1]}{\sup}|\ln (a-t)-p(t)|=\frac{1+o(1)}{L\sqrt{a^2-1}(a+\sqrt{a^2-1})^L}.
\end{eqnarray}
Let $x=1-(t+1)/(a+1)$, and $\eta=(a-1)/(a+1)$, then the above equation can be transformed to the following one:
\begin{eqnarray}
\underset{p\in \mathcal{P}_L}{\inf}\underset{x\in [\eta,1]}{\sup}|\ln x-p(x)|=\frac{1+o(1)}{L\frac{\sqrt{4\eta}}{1-\eta}\left(\frac{1+\eta}{1-\eta}+\frac{\sqrt{4\eta}}{1-\eta}\right)^L},
\end{eqnarray}
i.e. there exist two constants $c_1(\eta)$ and $c_2(\eta)$ that depend on $\eta$, such that
\begin{eqnarray}
\underset{p\in \mathcal{P}_L}{\inf}\underset{x\in [\eta,1]}{\sup}|\ln x-p(x)|\geq \frac{c_1(\eta)}{Lc_2^L(\eta)}.
\end{eqnarray}
Hence
\begin{eqnarray}
\Delta\geq \frac{2\alpha c_1(\eta)}{Lc_2^L(\eta)}.
\end{eqnarray}
To bound the total variation term in \eqref{eq:lecam}, we use the following lemma.
\begin{lem}\label{lem:TV}
	(\cite{wu2016minimax}, Lemma 3) Let $Z, Z'$ be random variables on $[0,A]$. If $\mathbb{E}[V^j]=\mathbb{E}[V'^j]$ for $j=1,\ldots,L$, and $L>2eA$, then
	\begin{eqnarray}
	\mathbb{TV}\left(\mathbb{E}[\text{Poi}(Z)],\mathbb{E}[\text{Poi}(Z')]\right)\leq \left(\frac{2eA}{L}\right)^L.
	\end{eqnarray}	
\end{lem}
Substitute $Z, Z'$ with $NU/m$ and $NU'/m$, and let $A=N\lambda/m$, we get
\begin{eqnarray}
\mathbb{TV}\left(\mathbb{E}\left[\text{Poi}\left(\frac{NU}{m}\right)\right], \mathbb{E}\left[\text{Poi}\left(\frac{NU'}{m}\right)\right]\right)\leq \left(\frac{2eN\lambda}{mL}\right)^L\leq \left(\frac{2eND^d}{L}\right)^L,
\label{eq:TVbound}
\end{eqnarray}
in which the last step holds because $\lambda\leq (1-\epsilon)mD^d$.

Let $L,D,m$ change in the following way:
\begin{eqnarray}
L&=&\left\lfloor \frac{\ln \ln N}{\ln c_2(\eta)}\right\rfloor,\label{eq:L}\\
D&=&\left(\frac{L}{2e}\right)^\frac{1}{d}N^{-\frac{1}{d}\left(1+\frac{1}{L}\right)}\label{eq:D},
\end{eqnarray}
and from \eqref{eq:fa},
\begin{eqnarray}
m\sim D^{-(d-1)}\sim L^{-\left(1-\frac{1}{d}\right)}N^{\left(1-\frac{1}{d}\right)\left(1+\frac{1}{L}\right)},\label{eq:m}
\end{eqnarray}
and
\begin{eqnarray}
\lambda\sim mD^d\sim L^\frac{1}{d}N^{-\frac{1}{d}\left(1+\frac{1}{L}\right)},\label{eq:lam}\\
\alpha=\lambda \eta\sim L^\frac{1}{d}N^{-\frac{1}{d}\left(1+\frac{1}{L}\right)}.\label{eq:alpha}
\end{eqnarray}
Then
\begin{eqnarray}
\Delta\geq \frac{2\alpha c_1(\eta)}{Lc_2^L(\eta)}\gtrsim \frac{\alpha}{\ln N\ln\ln N}.
\end{eqnarray}
Note that the second, third and fourth term in the bracket at the right hand side of \eqref{eq:lecam} converge to zero. In particular, for the second term,
\begin{eqnarray}
\frac{\lambda^2\left(\ln \frac{m}{\lambda}\right)^2}{m\Delta^2}\sim \frac{(\ln N)^4}{m}\rightarrow 0.
\end{eqnarray}
For the third term,
\begin{eqnarray}
m\mathbb{TV}\left(\mathbb{E}\left[\text{Poi}\left(\frac{NU}{m}\right)\right], \mathbb{E}\left[\text{Poi}\left(\frac{NU'}{m}\right)\right]\right)\leq \left(\frac{2eND^d}{L}\right)^L m=\frac{m}{N}\rightarrow 0,
\end{eqnarray}
and it is straightforward to show that the fourth term also converges to zero. Therefore, from Lemma \ref{lem:lecam},
\begin{eqnarray}
R_{a3}(N,M,\epsilon)\gtrsim \Delta^2\gtrsim L^\frac{2}{d}N^{-\frac{2}{d}\left(1+\frac{1}{L}\right)}\frac{1}{\ln^2 N\ln^2 \ln N}.
\end{eqnarray}

Pick $\eta$ such that $c_2(\eta)=e^2$. According to condition 1) in the statement of Lemma \ref{lem:lecam}, this is possible if $c$ is sufficiently small. Then
\begin{eqnarray}
R_{a3}(N,M,\epsilon)\gtrsim N^{-\frac{2}{d}\left(1+\frac{2}{\ln \ln N}\right)}\ln^{-2}N \ln^{-\left(2-\frac{2}{d}\right)}(\ln N).
\label{eq:R3lb}
\end{eqnarray}
From Lemma \ref{lem:2to3}, and note that from \eqref{eq:eps},
\begin{eqnarray}
\epsilon^2=\frac{16\lambda^2}{m^2}\sim \frac{m^2 D^{2d}}{m}\sim D^{d+1},
\end{eqnarray}
which converges sufficiently fast, thus $R_{a2}(N(1-\epsilon))$ can also be lower bounded with the right hand side of \eqref{eq:R3lb}. From \eqref{eq:0to1} and \eqref{eq:1to2}, 
\begin{eqnarray}
R_a(N,M)\gtrsim N^{-\frac{2}{d}\left(1+\frac{2}{\ln \ln N}\right)}\ln^{-2}N \ln^{-\left(2-\frac{2}{d}\right)}(\ln N).
\end{eqnarray}

\textbf{Proof of \eqref{eq:Ra3}}.

Define
\begin{eqnarray}
\mathcal{G}_a&=&\left\{(f,g)|f(\mathbf{x})=(1-\alpha)Q_a(\mathbf{x})+\sum_{i=1}^m \frac{\alpha}{mD^d}Q_a\left(\frac{\mathbf{x}-\mathbf{a}_i}{D}\right),\right.\nonumber\\ &&\hspace{15mm}g(\mathbf{x})=(1-\alpha)Q_a(\mathbf{x})
+\sum_{i=1}^m \frac{v_i}{mD^d}Q_a\left(\frac{\mathbf{x}-\mathbf{a}_i}{D}\right),\nonumber\\
&&\hspace{15mm}\left.\frac{1}{m}\sum_{i=1}^m v_i=\alpha, 1<mD^{d-1}<C_1, \frac{u_i}{mD^d}\in  (c,1)\right\}.\nonumber\\
\label{eq:ga}
\end{eqnarray}
Then for any $(f,g)\in \mathcal{G}_a$,
\begin{eqnarray}
D(f||g)=\sum_{i=1}^m \frac{\alpha}{m}\ln \frac{\alpha}{v_i}=\alpha \ln \alpha-\frac{\alpha}{m}\sum_{i=1}^m \ln v_i.
\label{eq:kl2}
\end{eqnarray}

Define
\begin{eqnarray}
R_{a4}(N,M)=\underset{\hat{D}}{\inf}\underset{(f,g)\in \mathcal{G}_a}{\sup}\mathbb{E}[(\hat{D}(N,M)-D(f||g))^2],
\end{eqnarray}
then for sufficiently large $U_g$ and sufficiently low $L_g$, we have $R_a(N,M)\geq R_{a4}(N,M)$. 

We use Poisson sampling again. Define
\begin{eqnarray}
R_{a5}(N,M)=\underset{\hat{D}}{\inf}\underset{(f,g)\in \mathcal{G}_a}{\sup}\mathbb{E}[(\hat{D}(N,M')-D(f||g))^2],
\end{eqnarray}
in which $M'\sim \text{Poi}(M)$. Then we have the following lemma.
\begin{lem}\label{lem:4to5}
	\begin{eqnarray}
	R_{a4}(N,M)\geq R_{a5}(N,2M)-\frac{1}{4}\alpha^2\ln^2 c\exp[-(1-\ln 2)M].
	\label{eq:4to5}
	\end{eqnarray}
\end{lem}
\begin{proof}
	Please refer to Appendix~\ref{sec:4to5}.
\end{proof}
Define
\begin{eqnarray}
\mathcal{G}_a'&=&\left\{(f,g)|f(\mathbf{x})=(1-\alpha)Q_a(\mathbf{x})+\sum_{i=1}^m \frac{\alpha}{mD^d}Q_a\left(\frac{\mathbf{x}-\mathbf{a}_i}{D}\right),\right.\nonumber\\ &&\hspace{15mm}g(\mathbf{x})=(1-\alpha)Q_a(\mathbf{x})
+\sum_{i=1}^m \frac{v_i}{mD^d}Q_a\left(\frac{\mathbf{x}-\mathbf{a}_i}{D}\right),\nonumber\\&&\hspace{15mm}\left.\left|\frac{1}{m}\sum_{i=1}^m v_i-\alpha\right|<\epsilon, 1<mD^{d-1}<C_1, \frac{u_i}{mD^d}\in  (c(1+\epsilon),1-\epsilon)\right\},\nonumber\\
\end{eqnarray}
and 
\begin{eqnarray}
R_{a6}(N,M)=\underset{\hat{D}}{\inf}\underset{(f,g)\in \mathcal{G}_a'}{\sup}\mathbb{E}[(\hat{D}(N,M')-D(f||g))^2],
\end{eqnarray}
in which $M'\sim \text{Poi}\left(M\int g(\mathbf{x})d\mathbf{x}\right)$. Then the following lemma lower bounds $R_{a5}$ with $R_{a6}$:
\begin{lem}\label{lem:5to6}
	If $\epsilon<\alpha/2$, then
	\begin{eqnarray}
	R_{a5}(N,(1-\epsilon)M)\geq \frac{1}{2}R_{a6}(N,M)-4\epsilon^2.
	\label{eq:5to6}
	\end{eqnarray}
\end{lem}
\begin{proof}
	Please refer to Appendix~\ref{sec:5to6}.
\end{proof}

Now we bound $R_{a6}(N,M,\epsilon)$ with the following lemma.
\begin{lem}\label{lem:lecam2}
	Let $V, V'$ be two random variables that satisfy the following conditions:
	
	(1) $V,V'\in [\eta \lambda,\lambda]$, in which $\lambda\leq (1-\epsilon)mD^d$, $0<\eta<1$ and $\eta \lambda\geq c(1+\epsilon)mD^d$;
	
	(2) $\mathbb{E}[V]=\mathbb{E}[V']=\alpha$.
	
	Define
	\begin{eqnarray}
	\Delta=|\mathbb{E}[\ln V]-\mathbb{E}[\ln V']|.
	\label{eq:Delta2}
	\end{eqnarray}
	
	Let $\epsilon=\lambda/\sqrt{m}$, then
	\begin{eqnarray}
	R_{a6}(N,M,\epsilon)\geq \frac{\alpha^2\Delta^2}{16}\left[\frac{1}{2}-\frac{8\ln^2 c}{m\Delta^2}-m\mathbb{TV}\left(\mathbb{E}\left[\text{Poi}\left(\frac{MV}{m}\right)\right],\mathbb{E}\left[\text{Poi}\left(\frac{MV'}{m}\right)\right]\right)\right].
	\label{eq:lecam2}
	\end{eqnarray}
\end{lem}
\begin{proof}
	Please refer to Appendix~\ref{sec:lecam2}.
\end{proof}

Now we use eq.(34) in \cite{wu2016minimax} again, which shows that there exist $V,V'\in [\eta\lambda,\lambda]$ that have matching moments up to $L$-th order, such that
\begin{eqnarray}
|\mathbb{E}[\ln V]-\mathbb{E}[\ln V']|=2\underset{p\in \mathcal{P}_L}{\inf}\underset{z\in [\eta,1]}{\sup}|\ln z-p(z)|.
\end{eqnarray} 
The remaining proof follows the proof of \eqref{eq:Ra2}. $L,D,m,\lambda$ and $\alpha$ take the same value as the equations from \eqref{eq:L} to \eqref{eq:alpha}, and then we can get similar bound as \eqref{eq:Ra2}, replacing $N$ with $M$.

\subsection{Proof of Lemma \ref{lem:1to2}}\label{sec:1to2}
Let $N'\sim \text{Poi}(2N)$, then
\begin{eqnarray}
R_{a2}(2N,M)&=&\underset{\hat{D}}{\inf}\underset{(f,g)\in \mathcal{F}_a}{\sup}\mathbb{E}\left[(\hat{D}(N,M)-D(f||g))^2\right]\nonumber\\
&\leq& \underset{\hat{D}}{\inf}\mathbb{E}\left[\underset{(f,g)\in \mathcal{F}_a}{\sup}\mathbb{E}\left[(\hat{D}(N,M)-D(f||g))^2|N'\right]\right]\nonumber\\
&=&\mathbb{E}\left[\underset{\hat{D}}{\inf}\underset{(f,g)\in \mathcal{F}_a}{\sup}\mathbb{E}\left[(\hat{D}(N,M)-D(f||g))^2|N'\right]\right]\nonumber\\
&=&\mathbb{E}[R_{a1}(N',M)]\nonumber\\
&=&\mathbb{E}[R_{a1}(N',M)|N'\geq N]\text{P}(N'\geq N)+\mathbb{E}[R_{a1}(N',M)|N'<N]\text{P}(N'<N),\nonumber\\
\label{eq:Ra2decomp}
\end{eqnarray}
in which the inequality in the second step comes from Jensen's inequality. Note that $R_{a1}(N,M)$ is a nonincreasing function of $N$, because if $N_1<N_2$, given $N_2$ samples $\{\mathbf{X}_1,\ldots, \mathbf{X}_{N_2} \}$, one can always pick $N_1$ samples for the estimation, thus $R_{a1}(N_1,M)\geq R_{a1}(N_2,M)$ always holds. Therefore
\begin{eqnarray}
\mathbb{E}[R_{a1}(N',M)|N'\geq N]\leq R_{a1}(N,M).
\label{eq:Ra2b1}
\end{eqnarray}
Moreover, since $N'\sim \text{Poi}(2N)$, use Chernoff inequality, we get
\begin{eqnarray}
\text{P}(N'<N)\leq \exp[-(1-\ln 2)N].
\label{eq:smalln}
\end{eqnarray}

Now it remains to bound $\mathbb{E}[R_{a1}(N',M)|N'\leq N]$. Note that we can always let the estimator be
\begin{eqnarray}
\hat{D}(f||g)=\frac{1}{2}\left(\underset{(f,g)\in \mathcal{F}_a}{\sup} D(f||g)+\underset{(f,g)\in \mathcal{F}_a}{\inf}D(f||g)\right),
\end{eqnarray}
hence
\begin{eqnarray}
\mathbb{E}[R_{a1}(N',M)|N'<N]\leq \frac{1}{4}\left(\underset{(f,g)\in \mathcal{F}_a}{\sup}D(f||g)-\underset{(f,g)\in \mathcal{F}_a}{\inf}D(f||g)\right)^2.
\end{eqnarray}
From the definition of $\mathcal{F}_a$ in \eqref{eq:fa}, for all $(f,g)\in \mathcal{F}_a$, 
\begin{eqnarray}
D(f||g)&=&\int f(\mathbf{x})\ln f(\mathbf{x})d\mathbf{x}-\int f(\mathbf{x})\ln g(\mathbf{x})d\mathbf{x}\nonumber\\
&=&-h(f)-\int f(\mathbf{x})\ln g(\mathbf{x})d\mathbf{x},
\end{eqnarray}
and
\begin{eqnarray}
\int f(\mathbf{x})\ln g(\mathbf{x})d\mathbf{x}&=&\int \sum_{i=1}^m \frac{u_i}{mD^d}Q_a\left(\frac{\mathbf{x}-\mathbf{a}_i}{D}\right)\ln \frac{\alpha}{mD^d v_d}d\mathbf{x}\nonumber\\
&=&\left(\frac{1}{m}\sum_{i=1}^m u_i\right)\ln \frac{\alpha}{mD^dv_d}\nonumber\\
&=&\alpha\ln \frac{\alpha}{mD^d v_d},
\label{eq:crossentropy}
\end{eqnarray}
which is the same for all $(f,g)\in \mathcal{F}_a$. In addition,
\begin{eqnarray}
h(f)&=&-\int f(\mathbf{x})\ln f(\mathbf{x})d\mathbf{x}\nonumber\\
&=&-(1-\alpha)\ln \frac{1}{v_d}-\frac{1}{m}\sum_{i=1}^m u_i\ln \frac{\alpha}{mD^dv_d}\nonumber\\
&=&(1-\alpha)\ln v_d+\alpha \ln (mD^d v_d)-\frac{1}{m}\sum_{i=1}^m u_i\ln u_i.
\label{eq:hf}
\end{eqnarray}
Hence,
\begin{eqnarray}
&&\mathbb{E}[R_{a1}(N',M)|N'<N]\nonumber\\
&\leq & \left(\underset{(f,g)\in \mathcal{F}_a}{\sup}h(f)-\underset{(f,g)\in \mathcal{F}_a}{\inf} h(f) \right)^2\nonumber\\
&=&\frac{1}{4}\left[\sup\left\{\frac{1}{m}\sum_{i=1}^m u_i\ln u_i|u_i>0,\frac{1}{m}\sum_{i=1}^m u_i=\alpha \right\}-\inf\left\{\frac{1}{m}\sum_{i=1}^m u_i\ln u_i|u_i>0,\frac{1}{m}\sum_{i=1}^m u_i=\alpha \right\}\right]^2\nonumber\\
&=&\frac{1}{4}\alpha^2 \ln^2\alpha\nonumber\\
&<&\frac{1}{4}.
\label{eq:maxR}
\end{eqnarray}
From \eqref{eq:Ra2decomp}, \eqref{eq:Ra2b1}, \eqref{eq:smalln} and \eqref{eq:maxR},
\begin{eqnarray}
R_{a2}(2N,M)\leq R_{a1}(N,M)+\frac{1}{4}\exp[-(1-\ln 2)N].
\end{eqnarray}
\subsection{Proof of Lemma \ref{lem:2to3}}\label{sec:2to3}
Recall that in \eqref{eq:fa2}, 
\begin{eqnarray}
f(\mathbf{x})=(1-\alpha)Q_a(\mathbf{x})+\frac{1}{q}\sum_{i=1}^m \frac{u_i}{mD^d}Q_a\left(\frac{\mathbf{x}-\mathbf{a}_i}{D}\right),
\end{eqnarray}
and $|(1/m)\sum_{i=1}^m u_i-\alpha|<\epsilon$. Define
\begin{eqnarray}
q=\frac{\sum_{i=1}^m u_i}{m\alpha},
\label{eq:q}
\end{eqnarray}
and
\begin{eqnarray}
f^*(\mathbf{x})=(1-\alpha)Q_a(\mathbf{x})+\frac{1}{q}\sum_{i=1}^m \frac{u_i}{mD^d}Q_a\left(\frac{\mathbf{x}-\mathbf{a}_i}{D}\right).
\end{eqnarray}
Then from \eqref{eq:fa2}, $|q-1|<\epsilon/\alpha$, $\int f^*(\mathbf{x})d\mathbf{x}=1$, and $f^*\in \mathcal{F}_a$. Hence
\begin{eqnarray}
R_{a3}(N,M,\epsilon)&=&\underset{\hat{D}}{\inf}\underset{(f,g)\in \mathcal{F}_a'}{\sup}\mathbb{E}\left[(\hat{D}(N,M)-D(f||g))^2 \right]\nonumber\\
&\leq & 2\underset{\hat{D}}{\inf}\underset{(f,g)\in \mathcal{F}_a}{\sup}\mathbb{E}\left[(\hat{D}(N,M)-D(f^*||g))^2 \right]+2\underset{(f,g)\in \mathcal{F}_a}{\sup}\left(D(f||g)-D(f^*||g)\right)^2\nonumber\\
&\leq & 2R_{a2}((1-\epsilon)N,M)+2\underset{(f,g)\in \mathcal{F}_a}{\sup}\left(D(f||g)-D(f^*||g)\right)^2.
\label{eq:cauchy}
\end{eqnarray}
Now we bound the second term.
\begin{eqnarray}
|D(f||g)-D(f^*||g)|\leq |h(f)-h(f^*)|+\left|\int f(\mathbf{x})\ln g(\mathbf{x})-\int f^*(\mathbf{x})\ln g(\mathbf{x})d\mathbf{x}\right|.
\end{eqnarray}
According to \eqref{eq:hf}, 
\begin{eqnarray}
|h(f)-h(f^*)|&=&\frac{1}{m}\left|\sum_{i=1}^m u_i\ln u_i-\sum_{i=1}^m \frac{u_i}{q}\ln \frac{u_i}{q}\right|\nonumber\\
&=&\frac{1}{m}\left|q\sum_{i=1}^m \frac{u_i}{q}\left(\ln \frac{u_i}{q}+\ln q\right)-\sum_{i=1}^m \frac{u_i}{q}\ln \frac{u_i}{q}\right|\nonumber\\
&\leq &\frac{1}{m}\left|(q-1)\sum_{i=1}^m \frac{u_i}{q}\ln \frac{u_i}{q}\right|+\frac{1}{m}\left|\sum_{i=1}^m u_i\ln q\right|\nonumber\\
&\overset{(a)}{\leq}&|1-q||\alpha\ln \alpha|+\alpha |q\ln q|\nonumber\\
&\overset{(b)}{\leq}& \epsilon\ln \frac{1}{\alpha}+\alpha\left(1+\frac{\epsilon}{\alpha}\right)\ln \left(1+\frac{\epsilon}{\alpha}\right)\nonumber\\
&\overset{(c)}{\leq}& \epsilon\ln \frac{1}{\alpha}+\frac{3}{2}\epsilon,
\end{eqnarray}
in which (a) is obtained by maximizing $|\sum_{i=1}^m (u_i/q)\ln (u_i/q)|$ under the restriction $(1/m)\sum_{i=1}^m (u_i/q)=\alpha$, (b) comes from $|q-1|<\epsilon/\alpha$, and (c) uses $\epsilon<\alpha/2$. Moreover,
\begin{eqnarray}
\left|\int f(\mathbf{x})\ln g(\mathbf{x})d\mathbf{x}-\int f^*(\mathbf{x})\ln g(\mathbf{x})d\mathbf{x}\right|&=&\left|\left(\frac{1}{m}\sum_{i=1}^m u_i-\alpha\right)\ln \frac{\alpha}{mD^dv_d}\right|\nonumber\\
&\leq & \epsilon\left|\ln \frac{\alpha}{mD^dv_d}\right|.
\end{eqnarray}
Hence
\begin{eqnarray}
|D(f||g)-D(f^*||g)|\leq \epsilon\left|\ln \frac{\alpha}{mD^dv_d}\right|+\epsilon\ln \frac{1}{\alpha}+\frac{3}{2}\epsilon.
\end{eqnarray}
Therefore
\begin{eqnarray}
R_{a3}(N,M,\epsilon)\leq 2R_{a2}((1-\epsilon)N,M)+6\epsilon^2\left(\ln^2 \frac{\alpha}{mD^d v_d}+\ln^2\alpha+\frac{9}{4}\right).
\end{eqnarray}
\subsection{Proof of Lemma \ref{lem:4to5}}\label{sec:4to5}
Similar to the proof of Lemma \ref{lem:1to2},
\begin{eqnarray}
R_{a5}(N,2M)\leq R_{a4}(N,M)+\exp[-(1-\ln 2)M]\mathbb{E}[R_{a4}(N,M')|M'<M],
\end{eqnarray}
and
\begin{eqnarray}
\mathbb{E}[R_{a4}(N,M')|M'<M]&\leq& \frac{1}{4}\left(\underset{(f,g)\in \mathcal{G}_a}{\sup} D(f||g)-\underset{(f,g)\in \mathcal{G}_a}{\inf}D(f||g)\right)^2\nonumber\\
&=&\frac{1}{4}\left(\frac{\alpha}{m}\sup\left\{\sum_{i=1}^m \ln v_i|v_i\in (cmD^d,mD^d), \frac{1}{m}\sum_{i=1}^m v_i=\alpha \right\}\right.\nonumber\\
&&\left.-\frac{\alpha}{m}\inf\left\{\sum_{i=1}^m \ln v_i|v_i\in (cmD^d,mD^d), \frac{1}{m}\sum_{i=1}^m v_i=\alpha \right\}\right)\nonumber\\
&\leq & \frac{1}{4}\alpha^2 \ln^2 c.
\end{eqnarray}
The proof is complete.
\subsection{Proof of Lemma \ref{lem:5to6}}\label{sec:5to6}
Similar to the proof of Lemma \ref{lem:2to3}, consider that
\begin{eqnarray}
g(\mathbf{x})=(1-\alpha)Q_a(\mathbf{x})+\frac{1}{mD^d}\sum_{i=1}^m v_iQ_a\left(\frac{\mathbf{x}-\mathbf{a}_i}{D}\right),
\end{eqnarray}
define $q=(\sum_{i=1}^m v_i)/(m\alpha)$, and
\begin{eqnarray}
g^*(\mathbf{x})=(1-\alpha) Q_a(\mathbf{x})+\frac{1}{q}\sum_{i=1}^m \frac{v_i}{mD^d}Q_a\left(\frac{\mathbf{x}-\mathbf{a}_i}{D}\right).
\end{eqnarray}
Similar to \eqref{eq:cauchy},
\begin{eqnarray}
R_{a6}(N,M,\epsilon)\leq 2R_{a5}(N,(1-\epsilon)M)+2\underset{(f,g)\in \mathcal{G}_a'}{\sup}\left(D(f||g)-D(f||g^*)\right)^2,
\end{eqnarray}
and
\begin{eqnarray}
|D(f||g)-D(f||g^*)|=\left|f(\mathbf{x})\ln \frac{g(\mathbf{x})}{g^*(\mathbf{x})}d\mathbf{x}\right|=\alpha|\ln q|\leq 2\epsilon,
\end{eqnarray}
in which the last step holds since $|q-1|<\epsilon/\alpha$ and $\epsilon<\alpha/2$. The proof is complete.

\subsection{Proof of Lemma \ref{lem:lecam2}}\label{sec:lecam2}
Let $g_1$, $g_2$ be two random functions:
\begin{eqnarray}
g_1(\mathbf{x})&=&(1-\alpha)Q_a(\mathbf{x})+\sum_{i=1}^m \frac{V_i}{mD^d}Q_a\left(\frac{\mathbf{x}-\mathbf{a}_i}{D}\right),\\
g_2(\mathbf{x})&=&(1-\alpha)Q_a(\mathbf{x})+\sum_{i=1}^m \frac{V_i'}{mD^d}Q_a\left(\frac{\mathbf{x}-\mathbf{a}_i}{D}\right).
\end{eqnarray}
Define two events:
\begin{eqnarray}
E&=&\left\{ \left|\frac{1}{m}\sum_{i=1}^m V_i-\alpha\right|\leq \epsilon, |D(f||g_1)-\mathbb{E}[D(f||g_1)]|\leq \frac{1}{4}\alpha\Delta \right\},\label{eq:Edef}\\
E'&=&\left\{ \left|\frac{1}{m}\sum_{i=1}^m V_i'-\alpha\right|\leq \epsilon, |D(f||g_2)-\mathbb{E}[D(f||g_2)]|\leq \frac{1}{4}\alpha\Delta \right\},\label{eq:E'def}
\end{eqnarray}
then
\begin{eqnarray}
\text{P}\left(\left|\frac{1}{m}\sum_{i=1}^m V_i-\alpha\right|>\epsilon\right)\leq \frac{\Var[V]}{m\epsilon^2}\leq \frac{\lambda^2}{4m\epsilon^2}=\frac{1}{4}.
\end{eqnarray}
Consider that $|\ln V|\in (\ln (1/\lambda),\ln(1/(\eta\lambda)))$, we have
\begin{eqnarray}
\Var[\ln V]\leq \frac{1}{4}\ln^2\eta\leq \frac{1}{4}\ln^2 c,
\end{eqnarray}
hence for $i=1,2$,
\begin{eqnarray}
\text{P}\left(|D(f||g_i)-\mathbb{E}[D(f||g_i)]|>\frac{1}{4}\alpha\Delta\right)&\leq &\frac{16}{\alpha^2\Delta^2}\Var[D(f||g_i)]\nonumber\\
&=&\frac{16}{\alpha^2\Delta^2 m}\Var[\alpha \ln V]\nonumber\\
&\leq & \frac{4\ln^2 c}{m\Delta^2}.
\end{eqnarray}
Therefore
\begin{eqnarray}
\max\{P(E^c), P(E'^c) \}\leq \frac{1}{4}+\frac{4\ln^2 c}{m\Delta^2}.
\end{eqnarray}
According to \eqref{eq:kl2},
\begin{eqnarray}
|\mathbb{E}[D(f||g_1)]-\mathbb{E}[D(f||g_2)]=\alpha |\mathbb{E}[\ln V]-\mathbb{E}[\ln V']|=\alpha\Delta.
\end{eqnarray}
From the definition of $E$, $E'$ in \eqref{eq:Edef} and \eqref{eq:E'def}, if $E,E'$ happen, then
\begin{eqnarray}
|D(f||g_1)-D(f||g_2)|\leq \frac{1}{2}\alpha \Delta.
\end{eqnarray}
Denote $\pi_1^*$ as the distribution of samples according to $g_1$ conditional on $E$, and $\pi_2^*$ as the distribution according to $g_2$ conditional on $E'$. Then under $\pi_1^*$, $\pi_2^*$,
\begin{eqnarray}
\mathbb{TV}(\pi_1^*,\pi_2^*)\leq \mathbb{TV}(\pi_1,\pi_2)+P(E^c)+P(E'^c),
\end{eqnarray}
and
\begin{eqnarray}
\mathbb{TV}(\pi_1,\pi_2)\leq m\mathbb{TV}\left(\mathbb{E}\left[\text{Poi}\left(\frac{MV}{m}\right)\right],\mathbb{E}\left[\text{Poi}\left(\frac{MV'}{m}\right)\right] \right).
\end{eqnarray}
Then according to Le Cam's lemma,
\begin{eqnarray}
R_{a6}(N,M,\epsilon)&\geq&\frac{1}{4}\left(\frac{1}{2}\alpha\Delta\right)^2(1-\mathbb{TV}(\pi_1^*,\pi_2^*))\nonumber\\
&\geq &\frac{\alpha^2\Delta^2}{16}\left[\frac{1}{2}-\frac{8\ln^2 c}{m\Delta^2}-m\mathbb{TV}\left(\mathbb{E}\left[\text{Poi}\left(\frac{MV}{m}\right)\right],\mathbb{E}\left[\text{Poi}\left(\frac{MV'}{m}\right)\right]\right)\right].\nonumber\\
\end{eqnarray}
The proof is complete.
\section{Proof of Theorem \ref{thm:mmx2}}\label{sec:mmx2}
Similar to Theorem \ref{thm:mmx1}, the proof can be divided into proving the following three bounds:
\begin{eqnarray}
R_b(N,M)&\gtrsim& \frac{1}{M}+\frac{1}{N};\label{eq:Rb1}\\
R_b(N,M)&\gtrsim& N^{-\frac{2\gamma}{d+2}}(\ln N)^{-\frac{4d+8-4\gamma}{d+2}};\label{eq:Rb2}\\
R_b(N,M)&\gtrsim& M^{-\frac{2\gamma}{d+2}}(\ln M)^{-\frac{4d+8-4\gamma}{d+2}}.\label{eq:Rb3}
\end{eqnarray}
\textbf{Proof of \eqref{eq:Rb1}}.

Let
\begin{eqnarray}
g(\mathbf{x})=\frac{1}{\sqrt{2\pi}}\exp\left[-\frac{1}{2}x_1^2\right],
\end{eqnarray}
in which $x_1$ is the value of the first coordinate of $\mathbf{x}$, and
\begin{eqnarray}
f_i(\mathbf{x})=\frac{1}{\sqrt{2\pi} \sigma_i} \exp\left[-\frac{x_1^2}{2\sigma_i^2}\right], i=1,2,
\end{eqnarray}
in which $\sigma_2^2=1/2$, and $\sigma_1=(1+\delta)\sigma_2$. Then
\begin{eqnarray}
D(f_1||g)&=&\frac{1}{2}(\sigma_1^2-1)-\ln \sigma_1,\\
D(f_2||g)&=&\frac{1}{2}(\sigma_2^2-1)-\ln \sigma_2,\\
\end{eqnarray}
and
\begin{eqnarray}
D(f_1||f_2)&=&\frac{1}{2}\left(\frac{\sigma_1^2}{\sigma_2^2}-1\right)-\ln \frac{\sigma_1}{\sigma_2}\nonumber\\
&=&\delta+\frac{1}{2}\delta^2-\ln (1+\delta)\nonumber\\
&\leq &\delta^2.
\end{eqnarray}
From Le Cam's lemma, 
\begin{eqnarray}
R_b(N,M)&\geq& \frac{1}{4}(D(f_2||g)-D(f_1||g))^2\exp[-ND(f_1||f_2)]\nonumber\\
&\geq & \frac{1}{4}\left(\ln(1+\delta)-\frac{1}{4}(2\delta+\delta^2)\right)^2\exp[-N\delta^2]\nonumber\\
&\geq & \frac{1}{4}\left(\frac{1}{2}\delta-\frac{3}{4}\delta^2\right)^2\exp[-N\delta^2].
\end{eqnarray}
Let $\delta=1/\sqrt{N}$, for sufficiently large $N$, $R_b(N,M)\geq 1/(32N)$. Similarly, let
\begin{eqnarray}
f(\mathbf{x})=\frac{1}{\sqrt{2\pi}}\exp\left[-\frac{1}{2}x_1^2\right],
\end{eqnarray}
and
\begin{eqnarray}
g_i(\mathbf{x})=\frac{1}{\sqrt{2\pi} \sigma_i} \exp\left[-\frac{x_1^2}{2\sigma_i^2}\right], i=1,2,
\end{eqnarray}
in which $\sigma_1=(1+\delta)\sigma_2$, then we can get $R_b(N,M)\gtrsim 1/M$. Hence
\begin{eqnarray}
R_b(N,M)\gtrsim \frac{1}{N}+\frac{1}{M}.
\end{eqnarray}

\textbf{Proof of \eqref{eq:Rb2}}.

To begin with, we construct $Q_b(\mathbf{x})$ that satisfies the following conditions:

(G1) $Q_b(\mathbf{x})$ is supported on $B(0,1)$, i.e. $Q_b(\mathbf{x})=0$ for $\norm{\mathbf{x}}>1$;

(G2) $\norm{\nabla^2 Q_b}\leq C_0$ for some constant $C_0$;

(G3) $\int_{B(0,1)}Q_b(\mathbf{x})d\mathbf{x}=1$;

(G4) $Q_b(\mathbf{x})\geq 0$ for all $\mathbf{x}$.

Let
\begin{eqnarray}
Q_m=\underset{\mathbf{x}}{\sup} Q_b(\mathbf{x}).
\end{eqnarray}

Define
\begin{eqnarray}
\mathcal{F}_b&=&\left\{(f,g)|f(\mathbf{x})=(1-\alpha)Q_b(\mathbf{x})+\sum_{i=1}^m \frac{u_i}{mD^d}Q_a\left(\frac{\mathbf{x}-\mathbf{a}_i}{D}\right),\right.\nonumber\\ &&\hspace{15mm}g(\mathbf{x})=(1-\alpha)Q_b(\mathbf{x})
+\sum_{i=1}^m \frac{\alpha}{mD^d}Q_b\left(\frac{\mathbf{x}-\mathbf{a}_i}{D}\right),\nonumber\\&&\hspace{15mm}\left.\frac{1}{m}\sum_{i=1}^m u_i=\alpha, 1<mD^{d+2(1-\gamma)}<C_1, \frac{u_i}{mD^{d+2}}<1 \right\}.\nonumber\\
\label{eq:fb}
\end{eqnarray}

In \eqref{eq:fb}, there are two conditions that are different from the definition of $\mathcal{F}_a$ in \eqref{eq:fa}: $1<mD^{d+2(1-\gamma)}<C_1$, and $u_i/(mD^{d+2})<1$. The first one is designed so that the distribution satisfies the tail assumption (Assumption 2 (b)). For $t\leq 1$,
\begin{eqnarray}
\text{P}(f(\mathbf{X})\leq t)&\leq& \left\{
\begin{array}{ccc}
tv_d+mtv_dD^d &\text{if} & t\leq D^2Q_m\\
tv_d+\alpha &\text{if} & t>D^2Q_m\\
\end{array}
\right.\nonumber\\
&\leq &tv_d+mD^{d+2(1-\gamma)}Q_m^{1-\gamma}v_dt^\gamma\nonumber\\
&\leq &\mu t^\gamma,
\end{eqnarray}
in which $\mu=v_d(1+C_1Q_m^{1-\gamma})$.

Follow the analysis in \cite{zhao2019analysis}, we can still get eq.(100) in \cite{zhao2019analysis}, i.e.
\begin{eqnarray}
R(N,M)\gtrsim \left(\frac{m}{N\ln m}\right)^2.
\end{eqnarray}

Let 
\begin{eqnarray}
D\sim N^{-\frac{1}{d+2}}(\ln N)^\frac{1}{d+2},
\end{eqnarray}
then
\begin{eqnarray}
m\sim D^{-d-2(1-\gamma)}\sim N^{\frac{d+2(1-\gamma)}{d+2}}(\ln N)^{-\frac{d+2(1-\gamma)}{d+2}}.
\end{eqnarray}
Hence
\begin{eqnarray}
R_b(N,M)\gtrsim N^{-\frac{4\gamma}{d+2}}(\ln N)^{-\frac{4d+8-4\gamma}{d+2}}.
\end{eqnarray}

\textbf{Proof of \eqref{eq:Rb3}}.
Define
\begin{eqnarray}
\mathcal{G}_b&=&\left\{(f,g)|f(\mathbf{x})=(1-\alpha)Q_b(\mathbf{x})+\sum_{i=1}^m \frac{\alpha}{mD^d}Q_a\left(\frac{\mathbf{x}-\mathbf{a}_i}{D}\right),\right.\nonumber\\ &&\hspace{15mm}g(\mathbf{x})=(1-\alpha)Q_b(\mathbf{x})
+\sum_{i=1}^m \frac{v_i}{mD^d}Q_b\left(\frac{\mathbf{x}-\mathbf{a}_i}{D}\right),\nonumber\\&&\hspace{15mm}\left.\frac{1}{m}\sum_{i=1}^m v_i=\alpha, 1<mD^{d+2(1-\gamma)}<C_1, \frac{v_i}{mD^{d+2}}<1,v_i\geq C_2\alpha \right\},\nonumber\\
\label{eq:gb}
\end{eqnarray}
in which $C_1$ and $C_2$ are two constants. Comparing with the definition of $\mathcal{F}_b$ in \eqref{eq:fb}, we add a new condition $v_i\geq C_2\alpha$, to ensure that $f/g$ is always bounded by $1/C_2$. Similar to Theorem \ref{thm:mmx1}, Let $V,V'\in [C_2\alpha,\lambda]$, $\lambda=\alpha/\eta$, $\lambda\leq mD^{d+2}$. Moreover, we still define $\Delta$ as was already defined in \eqref{eq:Delta2}. Then from Lemma \ref{lem:lecam2},
\begin{eqnarray}
R(N,M)\gtrsim \alpha^2\Delta^2\left[\frac{1}{2}-\frac{8\ln^2 c}{m\Delta^2}-m\mathbb{TV}\left(\mathbb{E}\left[\text{Poi}\left(\frac{MV}{m}\right)\right],\mathbb{E}\left[\text{Poi}\left(\frac{MV'}{m}\right)\right]\right)\right],
\end{eqnarray}
and from \eqref{eq:TVbound},

\begin{eqnarray}
\mathbb{TV}\left(\mathbb{E}\left[\text{Poi}\left(\frac{NU}{m}\right)\right], \mathbb{E}\left[\text{Poi}\left(\frac{NU'}{m}\right)\right]\right)\leq \left(\frac{2eM\lambda}{mL}\right)^L\leq \left(\frac{2eM\alpha}{mL\eta}\right)^L.
\end{eqnarray}
From Lemma 5 in \cite{wu2016minimax}, there exists two constants $c$, $c'$ such that
\begin{eqnarray}
\Delta=\underset{p\in \mathcal{P}_L}{\inf}\underset{z\in [cL^{-2},1]}{\sup}|\ln z-p(z)|\geq c'.
\end{eqnarray}

Let $L=2\lfloor \ln m\rfloor$, and
\begin{eqnarray}
\lambda=\frac{m\ln m}{e^2M},
\end{eqnarray}
\begin{eqnarray}
\alpha=\frac{m}{M\ln m},
\end{eqnarray}

then
\begin{eqnarray}
R_b(N,M)\gtrsim \left(\frac{m}{M\ln m}\right)^2.
\end{eqnarray}

With the restriction $1<mD^{1+2(1-\gamma)}<C_1$ and $\lambda\leq mD^{d+2}$, we have
\begin{eqnarray}
D\sim M^{-\frac{1}{d+2}}\ln^\frac{1}{d+2}M,
\end{eqnarray}
\begin{eqnarray}
m\sim D^{-d-2(1-\gamma)}\sim M^\frac{d+2(1-\gamma)}{d+2}(\ln M)^{-\frac{d+2(1-\gamma)}{d+2}},
\end{eqnarray}
hence
\begin{eqnarray}
R_b(N,M)\gtrsim M^{-\frac{4\gamma}{d+2}}(\ln M)^{-\frac{4d+8-4\gamma}{d+2}}.
\end{eqnarray}
\small \bibliography{macros,knn}
\bibliographystyle{ieeetran}
\end{document}